\documentclass[format=acmsmall, review=false, screen=true]{acmart}

\usepackage{booktabs} 

\usepackage{subfigure}

\usepackage{multirow}

\usepackage[ruled]{algorithm2e} 

\SetAlFnt{\small}
\SetAlCapFnt{\small}
\SetAlCapNameFnt{\small}
\SetAlCapHSkip{0pt}
\IncMargin{-\parindent}

\setcopyright{acmlicensed}

\begin{document}
\title{Krylov Subspace Approximation for Local Community Detection in Large Networks}

\author{Kun He}
\affiliation{%
  \institution{Huazhong University of Science and Technology}
  \city{Wuhan}
  \country{China}
   \postcode{430074}}
\email{brooklet60@hust.edu.cn, kh555@cornell.edu}
\author{Pan Shi}
\affiliation{%
  \institution{Huazhong University of Science and Technology}
  \city{Wuhan}
  \country{China}
   \postcode{430074}}
\email{panshi@hust.edu.cn}
\author{David Bindel}
\affiliation{%
 \institution{Cornell University}
 \city{Ithaca}
 \state{NY}
 \country{USA}
 \postcode{14853}}
\email{bindel@cs.cornell.edu}
\author{John E. Hopcroft}
\affiliation{%
  \institution{Cornell University}
  \city{Ithaca}
 \state{NY}
 \country{USA}
 \postcode{14853}}
\email{jeh@cs.cornell.edu}

\begin{abstract}
Community detection is an important information mining task to uncover modular structures in large networks. For increasingly common large network data sets, global community detection is prohibitively expensive, and attention has shifted to methods that mine local communities,~{\em i.e.}~identifying all latent members of a particular community from a few labeled seed members. To address such semi-supervised mining task, we systematically develop a local spectral subspace-based community detection method, called LOSP. We define a family of local spectral subspaces based on Krylov subspaces, and seek a sparse indicator for the target community via an $\ell_1$ norm minimization over the Krylov subspace. Variants of LOSP depend on type of random walks with different diffusion speeds, type of random walks, dimension of the local spectral subspace and step of diffusions. The effectiveness of the proposed LOSP approach is theoretically analyzed based on Rayleigh quotients, and it is experimentally verified on a wide variety of real-world networks across social, production and biological domains, as well as on an extensive set of synthetic LFR benchmark datasets.
\end{abstract}

%


%
%


\maketitle

\renewcommand{\shortauthors}{K. He \emph{et al}.}

\section{Introduction}
\label{sec:Introduction}
Community detection has arisen as one of the significant topics in network analysis and graph mining.
Many problems in information science, social science, biology and physics can be formulated as problems of community detection.
With the rapid growth of the network scale, however, exploring
the {\em global} community structure~\cite{Palla2005CFinder,Ahn2010LinkCommunities}
becomes prohibitively expensive in such networks with millions or billions of nodes.
While most of the time people are just interested in the local structure of the graph neighborhood.
Hence, attention has shifted to methods that mine local community
structure without processing the whole large network~\cite{JaewonICDM2012,IsabelKDD14,FreeRider2015,jeub2015think,he2016local,shi2017local}.

In many situations, instead of finding all communities in a large network in an unsupervised manner,
we just want to quickly find communities from a small set of members labeled by a domain expert.
For example, in political participation networks, one might discover
the membership of a political group from a few representative politicians~\cite{polarization2013};
sales websites might use clusters in co-purchase networks to generate
product recommendations for a customer based on previous purchases;
and starting from a few well-studied genes, biologists may seek functionally similar genes via genetic interaction networks.

Communities in real-world networks are often small, comprising dozens
or hundreds of members~\cite{Leskovec2008,han2015probabilistic}.
Intuitively, the latent members in these small communities should be very close to any seed members chosen from the target communities.
The {\em seed set expansion} approach to community
detection~\cite{Palla2005CFinder,lancichinetti2011OSLOM,Whang2013SSE}
starts from ``seed'' nodes, or labeled members of a target
community~\cite{JaewonICDM2012,kloster2014heat,IsabelKDD14},
and incrementally grows the set by locally optimizing a community
scoring function.

A common theme in seed set expansion methods is to diffuse probability
from the seeds.
PageRank~\cite{pagerank2006,IsabelKDD14}, heat kernel~\cite{kloster2014heat,Chung2013heat} and
local spectral approximation~\cite{He2015LOSP,li2015LEMON,LiHK_TKDD18} are three main techniques for the probability diffusion.
Among these, the local spectral method is a newly proposed technique
that exhibits high performance for the local community detection task.
Motivated by standard spectral clustering methods that find disjoint global
communities from the leading eigenvectors of the graph Laplacian,
local spectral algorithms can handle local communities by seeking a sparse indicator vector
that contains the seeds and lies in a local spectral subspace rather than in the global eigenspace.
Moreover, starting from different seed sets, local spectral algorithms can uncover overlapping communities.

In this paper, we propose a local spectral (LOSP) algorithm to identify the members of small communities~\cite{He2015LOSP}, and systematically build the LOSP family based on various Krylov subspace approximation. Theoretical analysis and empirical analysis are provided to show the soundness and effectiveness of the proposed approach.
LOSP differs from standard spectral clustering in two ways:

\begin{itemize}
\item
    \textbf{Handling overlapping communities}.
    Standard spectral methods partition a graph into disjoint communities
    by $k$-means clustering in a coordinate system defined by
    eigenvectors or by recursive spectral bisection.  We instead use
    $\ell_1$-norm optimization to search an (approximate) invariant
    subspace for an indicator vector for a sparse set containing the
    seeds.  Overlapping communities correspond to different seed sets.

\item
   \textbf{Defining a local spectral subspace.}
   To determine the local structure around the seeds of interest,
   we calculate a local approximate invariant subspace
   via a Krylov subspace associated with short random walks.

\end{itemize}

Starting from a few random seeds, we sample locally in the graph to get a comparatively small subgraph containing most of the
latent members, so that the follow-up membership identification can
focus on a local subgraph instead of the whole network.
The LOSP methods are then used to extract the local community
from the sampled subgraph via a Krylov subspace formed
by short-random-walk diffusion vectors.

Our main contributions include:
\begin{itemize}
\item
\textbf{Sampling to reduce the complexity}.
We sample locally in the graph as a pre-processing to get a much smaller subgraph containing most of the
latent members around the seeds. The follow-up calculation 
 has a low complexity while maintaining an accurate covering on the target community.
In this way, the proposed method is applicable to large-scale networks.

\item
\textbf{Building a rich LOSP family}.
We systematically develop a family of
local spectral methods based on random walk diffusion.
We thoroughly investigate a rich set of diffusion methods; and
we find that light lazy random walk, lazy random walk and personalized PageRank are robust for
different parameters, and outperform the standard random walk diffusions. In general, light lazy random walk performs slightly better.
We also study the regular random walk diffusions spreading out from the seeds and the inverse random walk diffusions ending around the seeds.

\item
\textbf{Extensive demonstration of LOSP}.
We provide a diverse set of computational experiments on 28 synthetic benchmark graphs and eight real-world social and biological networks
to show that the proposed methods are much more accurate than previous local spectral approach LEMON \cite{LiHK_TKDD18}, projected gradient descent algorithm PGDc-d \cite{LaarhovenJMLR16}, the well-known personalized PageRank diffusion
algorithm \texttt{pprpush}~\cite{pagerank2006} and heat kernel algorithm~\texttt{hk-relax}~\cite{kloster2014heat}.

\end{itemize}

This work is a significant extension and improvement based on our initial LOSP method~\cite{He2015LOSP}.
By removing some costly seed-strengthen preprocessing and iterative post-processing of the initial LOSP~\cite{He2015LOSP},
the current LOSP costs only one third of the running time of the initial LOSP after the sampling. 
We also improve the key component of LOSP based on the Krylov subspace, and extend it to an integrated LOSP family.

\section{Related Work}
\label{sec:relatedWork}
Communities are densely intra-linked components with sparser inter-connections,
typically defined by means of metrics like modularity~\cite{Newman06} or conductance~\cite{Ncut2000}.
There are many different research directions as well as approaches for the community detection~\cite{Fortunato2012survey,Papadopoulos2012survey,xie2013overlapping,Bae_TKDD17,HeHICODE18}.
Early works focused on global structure mining~\cite{Palla2005CFinder,Ahn2010LinkCommunities,coscia2012demon}
while an increasing body of recent work focuses on local structure mining~\cite{Mahoney2012Localspectral,IsabelKDD14,LaarhovenJMLR16}.
We will focus on seed set expansion methods that uncover a local
community from a few seed members; these were initially designed for
global community detection, but have since been extensively used in
local community mining. We could regard global community detection as an unsupervised method for finding global community structure, and regard local community detection as a semi-supervised approach for mining local community structure supervised by the seed set in the target community we want to find.

\subsection{Global Community Detection}
Global community detection is an important topic in data mining. Global community detection methods include modularity optimization \cite{newman2004fast,blondel2008fast,cao2016weighted}, stochastic block models \cite{abbe2017community}, nonnegative matrix factorization \cite{zhang2012overlapping,yang2013overlapping}, spectral methods \cite{Spectral2007,Newman13}, global seed set expansion \cite{Palla2005CFinder,SoundarajanTKDD15,TKDE16Whang}, and many other techniques \cite{cao2018incorporating,he2018network}. In
this subsection, we highlight a few lines of works in the literature that are related to our method.


\textbf{Spectral-based method.}
Spectral method is a main technique for global community detection. Von Luxburg \cite{Spectral2007} introduces the most common spectral clustering algorithms based on different graph Laplacians, and discusses the advantages and disadvantages of different spectral clustering algorithms in details.  Using spectral algorithms, Newman studies the three related problems of detection by modularity maximization, statistical inference or normalized-cut graph partitioning, and finds that with certain choices of the free parameters in the spectral algorithms the three problems are identical~\cite{Newman13}. As the classic spectral clustering method could only find disjoint communities, Zhang \emph{et al.} \cite{zhang2007identification} propose a new algorithm to identify overlapping communities in small networks by combining a generalized modularity function, an approximated spectral mapping of network nodes into Euclidean space and fuzzy $c-$means clustering. Ali and Couillet \cite{ali2017improved} propose a spectral algorithm based on the family of ``$\alpha$-normalized'' adjacency matrices $\mathbf{A}$, which is in the type of $\mathbf{D^{-\alpha}AD^{-\alpha}}$ with $\mathbf{D}$ the degree matrix, to find community structure in large heterogeneous networks.

\textbf{Global seed set expansion.}
Many global community detection algorithms are based on seed set expansion. Clique Percolation~\cite{Palla2005CFinder}, the most classic method, starts from maximal $k-$cliques and merges cliques sharing $k-1$ nodes to form a percolation chain. OSLOM~\cite{lancichinetti2011OSLOM}, starts with each node as the initial seed and optimizes a fitness function, defined as the probability of finding the cluster in a random null model, to join together small clusters into statistically significant larger clusters. Seed Set Expansion (SSE)~\cite{Whang2013SSE,TKDE16Whang} identifies overlapping communities by expanding different types of seeds by a personalized PageRank diffusion. DEMON~\cite{coscia2012demon} and another independent work in \cite{SoundarajanTKDD15} identify very small, tightly connected sub-communities, create a new network in which each node represents such a sub-community, and then identify communities in this meta-network. Belfin \emph{et al.} \cite{belfin2018overlapping} propose a strategy for locating suitable superior seed set by applying various centrality measures in order to find overlapping communities.

\subsection{Local Community Detection}
\textbf{Local seed set expansion.}
Random walks have been extensively adopted as a subroutine
for locally expanding the seed set~\cite{andersen2006communities},
and this approach is
observed to produce communities correlated highly to the
ground-truth communities in real-world
networks~\cite{Bruno2014separability}.
PageRank, heat kernel and local spectral diffusions are three main
techniques for probability diffusion.

Spielman and Teng use the degree-normalized personalized~\cite{pagerank2004}
PageRank (DN PageRank) with truncation of small values to expand a
starting seed.
DN PageRank has been used in several subsequent PageRank-based
clustering algorithms~\cite{andersen2006communities,JaewonICDM2012},
including the popular PageRank Nibble method~\cite{pagerank2006}.
However, a study evaluating different variations of PageRank finds
that standard PageRank yields better performance than
DN PageRank~\cite{IsabelKDD14}.

The heat kernel method provides another local graph diffusion.  Based
on a continuous-time Markov chain, the heat kernel diffusion involves
the exponential of a generator matrix, which may be
approximated via a series of expansion.
Chung \emph{et al.} have proposed
a local graph partitioning based on the heat kernel diffusion~\cite{Chung2007heat,chung2009local},
and a Monte Carlo algorithm to estimate the heat kernel process~\cite{Chung2013heat}.
Another approach is described in~\cite{kloster2014heat}, where the
authors estimate the heat kernel diffusion via coordinate relaxation
on an implicit linear system;
their approach uncovers smaller communities with substantially higher
$F_1$ measures than those found through
the personalized PageRank diffusion.

Spectral methods are often used to extract disjoint
communities from a few leading eigenvectors of
a graph Laplacian~\cite{Spectral2007,Kannan2000Spectral}.
Recently, there has been a growing interest in adapting
the spectral approach
to mine the local structure around the seed set.
Mahoney, Orecchia, and Vishnoi~\cite{Mahoney2012Localspectral} introduce
a locally-biased analogue
of the second eigenvector for extracting local properties of data
graphs near an input seed set
by finding a sparse cut,
and apply
the method
to semi-supervised image segmentation and local community extraction.
In~\cite{He2015LOSP,li2015LEMON,LiHK_TKDD18}, the authors introduce an algorithm to
extract the local community by seeking a
sparse vector from the local spectral subspaces using $\ell_{1}$ norm
optimization.
They apply a power method for the subspace iteration using a standard
random walk on a modified graph with a self loop on each node, which
we call the light lazy random walk.
They also apply a reseeding iteration to improve the detection accuracy.

\textbf{Bounding the local community.}
All seed set expansion methods need a stopping
criterion, unless the size of the target community is known in advance.
Conductance is commonly recognized as the best stopping criterion~\cite{JaewonICDM2012,kloster2014heat,Whang2013SSE,TKDE16Whang}.
Yang and Leskovec~\cite{JaewonICDM2012} provide
widely-used real-world datasets with labeled ground truth, and
find that conductance and triad-partition-ratio (TPR) are the two
stopping rules yielding the highest detection accuracy.
He \emph{et al}.~\cite{He2015LOSP} propose two new metrics, TPN and
nMod, and compare them with conductance, modularity and TPR; and they show
that conductance and TPN consistently outperform other metrics.
Laarhoven and Marchiori~\cite{LaarhovenJMLR16} study a continuous relaxation of conductance by investigating the relation of conductance
with weighted kernel $k$-means.

\textbf{Local seeding strategies.}
The seeding strategy is a key part of seed set expansion algorithms.
Kloumann and Kleinberg~\cite{IsabelKDD14}
argue that random seeds are superior to high degree seeds, and
suggest domain experts provide seeds with a diverse degree
distribution.
Our initial LOSP paper~\cite{He2015LOSP} compares
low degree, random, high triangle participation (number of triangles inside the community
containing the seed) and low escape seeds (judged by probability
retained on the seeds after short random walks), and find all four
types of seeds yield almost the same accuracy.
Our initial LOSP work shows that low degree seeds spread out the probabilities slowly and better preserve
the local information, and random seeds are similar to low degree
seeds due to the power law degree distribution.
High triangle participation seeds and low escape seeds follow another
philosophy: they choose seeds more cohesive to the target community.

\section{Preliminaries}
\label{sec:preliminaries}
\subsection{Problem Formulation}
The local community detection problem can be formalized as follows.
We are given a connected, undirected graph $G = (V,E)$ with $n$ nodes and $m$ edges.
Let $\mathbf{A} \in \{0,1\}^{n \times n}$ be the associated adjacency matrix, and $\mathbf{D}$ the diagonal matrix of node degrees.
Let $S$ be the seed set of a few exemplary members in the target ground-truth community, denoted by a set of nodes $T$ ($S \subset T$, $|T| \ll |V|$).
And let $\mathbf{s} \in \{0,1\}^{n \times 1}$ be a binary indicator vector representing the exemplary members in $S$.
We are asked to identify the remaining latent members in the target community $T$.

\subsection{Datasets}
\label{sec:datasets}
We consider four groups with a total of 28 synthetic datasets, five SNAP datasets in social, product, and collaboration domains, and three biology networks for a comprehensive evaluation on the proposed LOSP algorithms.

\subsubsection{LFR Benchmark}
\label{subsec:LFR}
For synthetic datasets, we use the LFR standard benchmark networks
proposed by Lancichinetti \emph{et al}.~\cite{2008LFRbenchmark,2009LFRbenchmark}.
The LFR benchmark graphs have a built-in community structure that simulates properties of real-world networks accounting for heterogeneity of node degrees and community sizes that follow power law distribution.

We adopt the same set of parameter settings used in \cite{xie2013overlapping}
and generate four groups with a total of 28 LFR benchmark graphs.
Table \ref{LFR summary} summarizes the parameter settings we used, among which the mixing parameter $\mu$ has a big impact on the network topology.
Parameter $\mu$ controls the average fraction of neighboring nodes that do not belong to any community for each node,
two ranges of typical community size, big and small, are provided by $b$ and $s$.
Each node belongs to either one community or $om$ overlapping communities, and the number of nodes in overlapping communities is specified by $on$.
A larger $om$ or $on$ indicates more overlaps that are harder for the community detection task.

For four groups of configurations based on the community size and $on$, we vary $om$ from 2 to 8 to get seven networks in each group,
denoted as: \textbf{LFR\_s\_0.1} for $\{s:[10,50], on=500\}$, \textbf{LFR\_s\_0.5} for $\{s:[10,50], on=2500\}$,
\textbf{LFR\_b\_0.1} for $\{b:[20,100], on=500\}$, and \textbf{LFR\_b\_0.5} for $\{b:[20,100], on=2500\}$.
The average conductance for four groups of datasets are 0.522, 0.746, 0.497 and 0.733, respectively. We see more overlapping on the communities (a bigger \emph{on}) leads to a higher conductance.

\begin{table}[htbp]
\caption{Parameters for the LFR benchmarks.}
\label{LFR summary}
\centering
\begin{tabular}{l|l}
\hline
\bf{Parameter}&\bf{ Description}\\
\hline
$n = 5000$ & number of nodes in the graph\\
$\mu = 0.3$ & mixing parameter\\
$\bar d = 10$ & average degree of the nodes\\
$d_{max} = 50$ & maximum degree of the nodes \\
$s:[10,50], b:[20,100]$ & range of the community size \\
$\tau_1 = 2$ & node degree distribution exponent\\
$\tau_2 = 1$ &community size distribution exponent \\
$om \in \{2,3...,8\}$ &overlapping membership\\
$on \in \{500, 2500\}$& number of overlapping nodes\\
\hline			
\end{tabular}		
\end{table}

\subsubsection{Real-world Networks}

We consider five real-world network datasets with labeled ground truth from the Stanford Network Analysis Project (SNAP)\footnote{http://snap.stanford.edu}
and three genetic networks with labeled ground truth from the Isobase website\footnote{http://groups.csail.mit.edu/cb/mna/isobase/}.
\setlength{\belowdisplayskip}{0pt} \setlength{\belowdisplayshortskip}{0pt}
\setlength{\abovedisplayskip}{0pt} \setlength{\abovedisplayshortskip}{0pt}
\begin{itemize}
	\item \textbf{SNAP}: The five SNAP networks, \textbf{Amazon, DBLP, LiveJ, YouTube, Orkut}, are in the domains of social, product, and collaboration~\cite{JaewonICDM2012}. For each network, we adopt the top 5000 annotated communities with the highest quality evaluated with 6 metrics \cite{JaewonICDM2012}: Conductance, Flake-ODF, FOMD, TPR, Modularity and CutRatio. Our algorithm adopts the popular metric of conductance to automatically determine the community boundary. To make a fair comparison, we choose four state-of-the-art baselines that also adopt conductance to determine the community boundary.
	\item \textbf{Biology}: The three genetic networks from the Isobase website describe protein interactions.
	\textbf{HS} describes these interactions in humans, \textbf{SC} in \textit{S. cerevisiae}, a type of yeast, and \textbf{DM} in \textit{D. melanogaster}, a type of fruit fly. Such networks are interesting as communities may correspond to different genetic functions.
\end{itemize}

Table~\ref{table:realdata_stats} summarizes the networks and their ground truth communities.
We calculate the average and standard deviation of the community sizes, and the average conductance, where low conductance gives priority to communities with dense internal links and sparse external links.
We also define and calculate the roundness of communities.

\begin{definition}
\textbf{Roundness of a subgraph}. The roundness of a subgraph $G' = (V', E')$ is the average shortest path among all pair-wise nodes divided by the longest shortest path in the subgraph.
\end{definition}

The roundness value $R$ is 1 for a clique, and $R = \frac{|V'|+1}{3(|V'|-1)} \approx \frac{1}{3}$ if the subgraph is a straight line.
Because large roundness value indicates a ``round'' subgraph and small roundness value indicates a ``long and narrow'' subgraph, the roundness reveals
some information on the topology structure of the subgraph. Table~\ref{table:realdata_stats} shows that communities in the above real-world networks have an average roundness of about 0.67.
If we normalize the roundness value from [1/3, 1] to [0,1], then we get $R_{norm} = \frac{R - 1/3}{1 - 1/3}$.

\begin{table}[htbp]
\caption{Statistics for real-world networks and their ground truth communities.}
\label{table:realdata_stats}
\centering
\scalebox{0.8}{
\begin{tabular}{l | l  r r  | r  c  c c }
\hline
        \multirow{2}{*}{\bf{Domain}}    & \multicolumn{3}{c}{\bf{Network}} \vline& \multicolumn{4}{c}{\bf{Ground truth communities}} \\
 &\bf{Name} &\bf{\# Nodes}&\bf{\# Edges} &\multicolumn{1}{c}{\bf{Avg. $\pm$ Std. Size }} & \bf{Avg. Cond.} & \bf{Roundness $R$} & \bf{ $R_{norm}$}\\
 \hline
\bf{Product} &\bf{Amazon}     & 334,863    & 925,872    &13 $\pm$ 18   & 0.07  &    0.69 & 0.54\\
\bf{Collaboration}& \bf{DBLP} & 317,080    &1,049,866   &22 $\pm$ 201  & 0.41  &    0.74 & 0.61\\
\bf{Social} &\bf{LiveJ}       & 3,997,962  & 34,681,189 &28 $\pm$ 58   & 0.39  &    0.65 & 0.48\\
\bf{Social} &\bf{YouTube}     & 1,134,890  & 2,987,624  &21 $\pm$ 73   & 0.84  &    0.69 & 0.54\\
\bf{Social}& \bf{Orkut}       & 3,072,441  & 117,185,083&216 $\pm$ 321 & 0.73  &    0.50 & 0.25 \\
\hline
\bf{Biology}& \bf{DM}         & 15,294     & 485,408    &440 $\pm$ 2096 &  0.88 &    0.68 & 0.52\\
\bf{Biology}& \bf{HS}         & 10,153     & 54,570     &113 $\pm$ 412 &  0.88 &    0.58 & 0.37\\
\bf{Biology}& \bf{SC}         & 5,523      & 82,656     &67 $\pm$ 110 &  0.90 &    0.63 & 0.45\\
\hline
\end{tabular}}
\end{table}

\subsection{Evaluation Metric}
For the evaluation metric, we adopt $F_1$ score to quantify the similarity between the detected local community $C$ and the target ground truth community $T$. The $F_1$ score for each pair of $(C,T)$ is defined by:
\begin{equation*}
F_1(C,T) = \frac{2|C \cap T|}{|C|+|T|}.
\end{equation*}

\section{Krylov Subspace Clustering}
\label{sec:Algorithm}
It is well known that there is a close relationship between the combinatorial characteristics of a graph and the algebraic properties of its associated matrices~\cite{chung1997spectral}.
In this section, we present a local community detection method of finding a linear sparse coding on the Krylov subspace, which is a local approximation of the spectral subspace.
 We first present the local sampling method starting from the seeds to reach a comparatively small subgraph $G_s$, and provide the theoretical base that finding a local community containing the seeds corresponds to
finding a sparse linear coding in the invariant subspace spanned by the dominant eigenvectors of the transition matrix.
Then we propose a family of local spectral subspace definitions based on various short random walk diffusions,
and do local community detection by finding a sparse relaxed indicator vector that lies in a local spectral subspace representing the subordinative probability of the corresponding nodes.

\vspace{-0.5em}
\subsection{Local Sampling}

We first do a local sampling from the seeds for large networks to reduce the computational complexity.
According to the small world phenomenon and ``six degrees of separation'',
most members should be at most two or three steps far away from the seed members if we want to identify a small community of size hundreds.
Thus, we mainly use a few steps of Breadth-First Search (BFS) to expand nodes for the sampling method. To avoid BFS to expand too fast, we
adopt a strategy to filter some very popular nodes during the BFS expansion. The detailed description for the sampling method is as follows.

Starting from each seed, we do a one-round BFS. The $BFS(\cdot)$ operation achieves a set of nodes containing the source node and its neighbor nodes by breadth-first search.
For the current subgraph expanded by the BFS, we define the inward ratio of a node as the fraction of inward edges to the out-degree.
If the sampled subgraph is no greater than the lower bound $N_1$, then we do one more round of BFS to contain more neighbor nodes.
To avoid the next round of BFS to expand too fast, we first filter low inward ratio frontier nodes which may contain some inactive nodes or very popular nodes.
The $Filter(\cdot)$ operation chooses high inward ratio nodes until the total out-degree is no less than 3000.
In the end, we union all BFS subgraphs obtained from each seed.
If the amalgamated subgraph scale is larger than the upper bound $N_2$, we then conduct a $k$-step short random walk from the seeds to remove some low probability nodes. Details of the sampling are as shown in Algorithm \ref{alg:one}.
\begin{algorithm}
\SetAlgoNoLine	
\KwIn{Graph $G = (V,E)$, seed set $S \subseteq V$, lower bound of sampled size from each seed $N_1$, upper bound of the subgraph size $N_2$, upper bound of BFS rounds $t$, and steps of random walks $k$ for postprocessing}
\KwOut{Sampled subgraph $G_s = (V_s,E_s)$}
$V_s \leftarrow S$ \\
\For{each $s_i \in S$}
{
$V_i \leftarrow BFS(s_i)$ \\
$V'_i \leftarrow V_i$ \\
\While{($|V_i| < N_1$ and BFS rounds $\leq t$)}
{
$V'_i \leftarrow Filter(V_i')$ \\
$V_i' \leftarrow BFS(V'_i)$ \\
$V_i = V_i \cup V_i'$ \\
}
$V_s = V_s \cup V_i$ \\
}
$G_s = (V_s,E_s)$ is the induced subgraph from $V_s$ \\
\If{$|V_s| > N_2$}
{
Conduct a $k$-steps of random walk from $S$ in $G_s$ \\
$V_s \leftarrow$ \text{$N_2$ nodes with higher probability} \\
$G_s = (V_s,E_s)$ is the induced subgraph from $V_s$ \\
}
\caption{Sampling}
\label{alg:one}
\end{algorithm}

Denote the sampled subgraph as $G_s = (V_s,E_s)$ with $n_s$ nodes and $m_s$ edges in the following discussion.
We then identify the local community from this comparatively small subgraph instead of the original large network.
The complexity is only related to the degrees of the nodes and it is very quick in seconds for the datasets we considered.
The sampling quality, evaluated by the coverage ratio of the labeled nodes, plays a key role for the follow-up membership identification.
This pre-processing procedure significantly reduces the membership identification cost.

\subsection{Spectra and Local Community}
In this subsection, we provide the necessary theoretical base that finding a low-conductance community corresponds to finding a sparse indicator vector in the span of dominant eigenvectors of the transition matrix with larger eigenvalues.

Let $\mathbf{L} = \mathbf{D_s} - \mathbf{A_s}$ be the Laplacian matrix of $G_s$ where $\mathbf{A_s}$ and $\mathbf{D_s}$ denote the adjacency matrix and the diagonal degree matrix of $G_s$.
We define two normalized graph Laplacian matrices:
$$\mathbf{L_{rw}} = \mathbf{D_s^{-1}L} = \mathbf{I} - \mathbf{N_{rw}},~~ \mathbf{L_{sym}} =\mathbf{D_s^{-\frac{1}{2}}LD_s^{-\frac{1}{2}}} = \mathbf{I} - \mathbf{N_{sym}},
$$
\noindent where $\mathbf{I}$ is the  identity matrix, $\mathbf{N_{rw}}=\mathbf{D_s}^{-1}\mathbf{A_s}$ is the transition matrix,
and $\mathbf{N_{sym}}$ $=\mathbf{D_s^{-\frac{1}{2}}A_sD_s^{-\frac{1}{2}}}$ is the normalized adjacency matrix.

For a community $C$, the conductance~\cite{Ncut2000} of $C$ is defined as
\begin{equation*}
\Phi(C) = \frac{\text{cut}(C,\overline{C})}{\text{min}\{\text{vol}(C),\text{vol}(\overline{C})\}},
\end{equation*}
where $\overline{C}$ consists of all nodes outside $C$, cut$(C,\overline{C})$ denotes the number of edges between $C$ and $\overline{C}$,
and vol$(\cdot)$ calculates the ``edge volume'', i.e. for the subset nodes, we count their total node degrees in graph $G_s$.
Low conductance gives priority to a community with dense internal links and sparse external links.

Let $\mathbf{y} \in \{0,1\}^{n_s\times 1}$ be a binary indicator vector representing a small community $C$ in the sampled graph $G_s = (V_s,E_s)$.
Here for ``small community'', we mean  $\text{vol}(C)\leq \frac{1}{2}\text{vol}(V_s)$.
As $\mathbf{y^TD_sy}$ equals the total node degrees of $C$, and $\mathbf{y^TA_sy}$ equals two times the number of internal edges of $C$,
the conductance $\Phi(C)$ could be written as a generalized Rayleigh quotient
\begin{equation*}
\Phi(C) = \frac{\mathbf{y^TLy}}{\mathbf{y^TD_sy}} = \frac{\mathbf{(D_s^{\frac{1}{2}}y)^TL_{sym}(D_s^{\frac{1}{2}}y)}}{\mathbf{(D_s^{\frac{1}{2}}y)^T(D_s^{\frac{1}{2}}y)}}.
\label{eq:Rayleigh}
\end{equation*}

\begin{theorem}
\label{th:cheeger}
Let $\lambda_2$ be the second smallest eigenvalue of $\mathbf{L_{sym}}$, then conductance $\Phi(C)$ of a small community $C$ in graph
$G_s = (V_s,E_s)$ (``small'' means $\text{vol}(C) \leq 0.5\text{vol}(V_s)$) is bounded by
\[
\frac{\lambda_2}{2} \leq \Phi(C) \leq 1,
\]
where $\text{vol}(C)$ denotes for all nodes inside $C \subseteq V_s$, we count the total degree in graph $G_s$.
\end{theorem}
The proof omits here, and we attach the details in Appendix \ref{sec:Bounding}.

Let $\mathbf{L_{sym}} = \mathbf{Q \Lambda Q^T}$ be the eigendecomposition, where $\mathbf{Q} =
\begin{bmatrix}
\mathbf{q}_1 ~|\cdots|~ \mathbf{q}_{n_s}
\end{bmatrix}
$ is an orthonormal matrix and $\mathbf{\Lambda} = \operatorname{diag}(\lambda_1, \lambda_2, ... ,\lambda_{n_s})$, $\lambda_1 \leq \lambda_2 \leq ... \leq \lambda_{n_s}$.
Then
\begin{equation*}
\Phi(C) = \frac{\mathbf{(Q^TD_s^{\frac{1}{2}}y)^T \Lambda (Q^TD_s^{\frac{1}{2}}y)}}{\mathbf{(Q^TD_s^{\frac{1}{2}}y)^T(Q^TD_s^{\frac{1}{2}}y)}}.
\label{eq:Rayleigh2}
\end{equation*}

Let $x_i = \mathbf{q}_i^\mathbf{T}\mathbf{D_s^{\frac{1}{2}}y}$ be the projection of $\mathbf{D_s^{\frac{1}{2}}y}$ on the $i$th eigenvector $\mathbf{q}_i$ of $\mathbf{L_{sym}}$, we have

\begin{equation}
\Phi(C) = \frac{\sum_{i=1}^{n_s} \lambda_i x_i^2 }{\sum_{i=1}^{n_s}  x_i^2} = \sum_{i=1}^{n_s}  w_i \lambda_i,
\label{eq:Quadratic}
\end{equation}
where $w_i = \frac{x_i^2}{{\sum_{i=1}^{n_s}  x_i^2}}$ is the weighting coefficient of the eigenvalues. If $\Phi(C)$ is close to the smallest eigenvalue $\lambda_1$, then most of the weight on average must be on the eigenvalues close to $\lambda_1$.

\begin{theorem}
\label{th:conductance}
Let $\epsilon$ be a small positive real number, if $\Phi(C) < \lambda_1 + \epsilon$, then for any positive real number $t$,
\begin{equation*}
\sum_{i:\lambda_i<\lambda_1+t\epsilon} w_i > 1 - \frac{1}{t}.
\end{equation*}
\end{theorem}
\vspace{-0.2em}
\begin{proof}
By Eq. (\ref{eq:Quadratic}) we have
\begin{equation}
\begin{aligned}
\Phi(C) & = \sum_{i:\lambda_i<\lambda_1+t\epsilon} w_i\lambda_i + \sum_{j:\lambda_j\geq \lambda_1+t\epsilon} w_j\lambda_j \\
                & \geq \lambda_1 \sum_{i:\lambda_i<\lambda_1+t\epsilon} w_i + (\lambda_1+t\epsilon)\sum_{j:\lambda_j\geq \lambda_1+t\epsilon} w_j \\
                & = \lambda_1 + t\epsilon\sum_{j:\lambda_j\geq \lambda_1+t\epsilon} w_j. \notag
\end{aligned}
\end{equation}
As $\Phi(C) < \lambda_1+\epsilon$, we get
\[
\lambda_1 + t\epsilon\sum_{j:\lambda_j\geq \lambda_1+t\epsilon} w_j <  \lambda_1+\epsilon.
\]

Therefore,
\[
\sum_{j:\lambda_j\geq \lambda_1+t\epsilon} w_j < \frac{1}{t}, ~~
\sum_{i:\lambda_i<\lambda_1+t\epsilon} w_i > 1 - \frac{1}{t}.
\]
\end{proof}

Note that $\lambda_1 = 0$ for the Laplacian matrix $\mathbf{L_{sym}}$ \cite{Spectral2007}, however, Theorem~\ref{th:conductance} holds for any real number $\lambda_1$.
Theorem~\ref{th:conductance} indicates for a low conductance $\Phi(C)$ close to the smallest value $\lambda_1$, the smaller eigenvalues of $\mathbf{L_{sym}}$ contribute most of the weights.
As $w_i = \frac{x_i^2}{{\sum_{i=1}^{n_s}  x_i^2}}$, if we want to find a low conductance $\Phi(C)$, $w_i$ and hence $x_i$ should be larger on smaller eigenvalues $\lambda_i~(1\leq i \leq k \ll n_s)$.
As $x_i = \mathbf{q}_i^\mathbf{T}\mathbf{D_s^{\frac{1}{2}}y}$, a larger $x_i$ indicates a smaller angle between $\mathbf{D_s^{\frac{1}{2}}y}$ and eigenvector $\mathbf{q}_i$.
When we relax $\mathbf{y}$ from $\{0,1\}^{n_s\times 1}$ to $[0,1]^{n_s\times 1}$, the relaxed scaled indicator vector $\mathbf{D_s^{\frac{1}{2}}y}$ should be well approximated by a linear combination of the dominant eigenvectors with smaller eigenvalues.
As
\begin{equation*}
\mathbf{L_{rw}v} = \lambda \mathbf{v} \quad \Leftrightarrow \quad  \mathbf{L_{sym}(\mathbf{D_s^{\frac{1}{2}}v})} = \lambda (\mathbf{D_s^{\frac{1}{2}}v}),
\end{equation*}
where $\bf{v}$ is a nonzero vector. It shows that the relaxed indicator vector $\mathbf{y}$ should be well approximated by a linear combination of the dominant eigenvectors of $\mathbf{L_{rw}}$ with smaller eigenvalues.
Also,
\begin{equation}\label{eq:Nrw}
\mathbf{L_{rw}v} = \mathbf{(I - N_{rw})v}  = \lambda \mathbf{v}
\quad \Leftrightarrow \quad
\mathbf{N_{rw}v} = (1 - \lambda) \mathbf{v},
\end{equation}
it follows that $\mathbf{L_{rw}}$  and $\mathbf{N_{rw}}$ share the same set of eigenvectors and the corresponding eigenvalue  of $\mathbf{N_{rw}}$ is $1 - \lambda$ where $\lambda$ is the eigenvalue of $\mathbf{L_{rw}}$.
Equivalently, the relaxed indicator vector $\mathbf{y}$ should be well approximated by a linear combination of the eigenvectors of $\mathbf{N_{rw}}$ with larger eigenvalues.

This leads to our idea of finding a relaxed sparse indicator vector $\mathbf{y}$ containing the seeds in the span of the dominant eigenvectors with larger eigenvalues of $\mathbf{N_{rw}}$.

\begin{equation}
\begin{aligned}
\min ~~ & \Vert\bf{y}\Vert_1 = \mathbf{e^Ty} \\
\emph{s.t.}& ~(1)~ \mathbf{ y = Vu},  \\
           & ~(2)~ y_i \in [0,1],~ \label{eq:oneNorm} \\
           & ~(3)~ y_i \geq \frac{1}{|S|}, i \in S,
\end{aligned}
\end{equation}
where the column vectors of $\mathbf{V}$ are formed by the dominant eigenvectors of $\mathbf{N_{rw}}$ with larger values and $\mathbf{e}$ is the vector of all ones. $\bf{u}$ is the coefficient vector, and $y_i$ denotes the $i$th element of the relaxed indicator vector $\bf{y}$.
The constraint $\mathbf{y = Vu}$ indicates that $\mathbf{y}$ lies in the eigenspace spanned by the column vectors of $\mathbf{V}$.
Vector $\mathbf{y}$ indicates the local community with low conductance containing the labeled seeds.

\subsection{Krylov Subspace Approximation}
\subsubsection{Variants of Random Walk Diffusion}
Instead of using the eigenvalue decomposition on $\mathbf{N_{rw}}$, we consider short random walks for the probability diffusion starting from the seed set to get the ``local spectral subspace''.
We define several variants of the spectral diffusion based on different transition matrices for the random walks.

1) \textit{Standard Random Walk (SRW)} uses the transition matrix $\mathbf{N_{rw}}$ for the probability diffusion.
\begin{equation}\label{Eq1}
\mathbf{N_{rw}}=\mathbf{D_s}^{-1}\mathbf{A_s}.
\end{equation}

2) \textit{Light Lazy Random Walk (LLRW)} retains some probability at the current node for the random walks.
\begin{equation}
\mathbf{N_{rw}}=(\mathbf{D_s +\alpha \mathbf{I}})^{-1}( \alpha \mathbf{I} +\mathbf{A_s}),
\label{eq:lightlazyRW}
\end{equation}
where $\alpha \in N^{0+}$.  $\alpha= 0$ degenerates to the standard random walk and $\alpha = 1,2,3,...$ corresponds to a random walk in the modified graph with $1,2,3,...$ loops at each node.

3) \textit{Lazy Random Walk (LRW)} is defined by
\begin{equation}
\begin{aligned}
\mathbf{N_{rw}}&=(\mathbf{D_s +\alpha \mathbf{D_s}})^{-1}( \alpha \mathbf{D_s} +\mathbf{A_s})  \\
           & = \frac{\alpha}{1+\alpha} \mathbf{I} + \frac{1}{1+\alpha} \mathbf{D_s}^{-1}\mathbf{A_s},
           \label{eq:lazyRW}
\end{aligned}
\end{equation}
where $\alpha \in [0,1]$.  E.g. $\alpha = 0.1$ corresponds to a random walk that always retains $\frac{0.1}{1 + 0.1}$ probability on the current node during the diffusion process.
$\alpha= 0$ degenerates to the standard random walk.

4) \textit{Personalized PageRank (PPR)} is defined by
\begin{equation}
\mathbf{N_{rw}}=\alpha \mathbf{S} + ( 1- \alpha)\mathbf{D_s}^{-1}\mathbf{A_s},
\label{eq:per}
\end{equation}
where $\alpha \in [0,1]$ and $\mathbf{S}$ the diagonal matrix with binary indicators for the seed set $S$.
E.g. $\alpha = 0.1$ corresponds to a random walk that
always retains 10\% 
of the probability on the seed set. $\alpha= 0$ is the standard random walk.

\subsubsection{Regular and Inverse Random Walks}
Based on the above random walk diffusion definition, one step of random walk is defined as $\mathbf{N_{rw}^Tp}$ for a probability column vector $\mathbf{p}$, and the probability density for a random walk of length $k$ is given by a Markov chain:
\begin{equation} \label{eq:PonNrwT}
\mathbf{p}_k = \mathbf{N_{rw}^Tp}_{k-1} = \mathbf{(N_{rw}^T)}^k\mathbf{p}_0,
\end{equation}
\noindent where $\mathbf{p}_0$ is the initial probability density evenly assigned on the seeds.

We can also define an ``inverse random walk'':
\begin{equation}\label{eq:PonNrw}
\mathbf{p}_k = \mathbf{N_{rw}p}_{k-1} = \mathbf{(N_{rw})}^k\mathbf{p}_0.
\end{equation}
Here $\mathbf{p}_k$ indicates a probabilty density such that the probability concentrates to the seed set as $\mathbf{p}_0$ after $k$ steps of short random walks.
The value of $\mathbf{p}_k$ also shows a snapshot of the probability distribution for the local community around the seed set, and follow-up experiments also demonstrates the effectiveness of the ``inverse random walk'', which has a slightly lower accuracy as compared with the ``regular random walk''.

\subsubsection{Local Spectral Subspace}
Then we define a local spectral subspace as a proxy of the invariant subspace spanned by the
leading eigenvectors of $\mathbf{N_{rw}}$.
The local spectral subspace is defined on an order-$d$ Krylov matrix and $k$ is the number of diffusion steps:
\begin{equation}\label{eq:LOSP2}
\mathbf{V}_d^{(k)} = \mathbf{[}\mathbf{p}_k,\mathbf{p}_{k+1},...,\mathbf{p}_{k+d-1} \mathbf{]}.
\end{equation}
Here $k$ and $d$ are both some modest numbers. Then the Krylov subspace spanned by the column vectors of $\mathbf{V}_d^{(k)}$ is called the local spectral subspace, denoted by $\mathcal{V}_d^{(k)}$.

In the following discussion, we provide some theoretical analysis to relate the spectral property to the local spectral subspace $\mathcal{V}_d^{(k)}$.

\begin{lemma}
Let $G_s = (V_s,E_s)$ be a connected and non-bipartite graph with $n_s$ nodes and $m_s$ edges, $\mathbf{N_{rw}}$ (defined by Eq. (\ref{Eq1})) the transition matrix of $G_s$ with eigenvalues $\sigma_1 \geq \sigma_2 \geq ... \geq \sigma_{n_s}$, and the corresponding normalized eigenvectors are $\mathbf{u}_1 , \mathbf{u}_2, \ldots, \mathbf{u}_{n_s}$. Then $\mathbf{u}_1 , \mathbf{u}_2, \ldots, \mathbf{u}_{n_s}$ are linearly independent, and
\[
1 = \sigma_1 > \sigma_2 \geq ... \geq \sigma_{n_s} > -1, ~~ \mathbf{u}_1 = \frac{\mathbf{e}}{\Vert \mathbf{e} \Vert_2},
\]
where $\mathbf{e}$ is a vector of all ones.
\label{lemmaone}
\end{lemma}

\begin{proof}
By Eq. (\ref{eq:Nrw}), we know that $\mathbf{L_{rw}}$ and $\mathbf{N_{rw}}$ share the same eigenvectors $\mathbf{u}_1 , \mathbf{u}_2, \ldots, \mathbf{u}_{n_s}$ and the corresponding eigenvalues of $\mathbf{L_{rw}}$ are $\lambda_1 \leq \lambda_2 \leq ... \leq \lambda_{n_s}$ where $\lambda_i = 1 - \sigma_i~(1 \leq i \leq n_s)$.

According to Proposition 3 of \cite{Spectral2007}, $\mathbf{L_{sym}}$ and $\mathbf{L_{rw}}$ share the $n_s$ non-negative eigenvalues $0 \leq \lambda_1 \leq \lambda_2 \leq ... \leq \lambda_{n_s}$ and the corresponding eigenvectors of $\mathbf{L_{sym}}$ are $\mathbf{D_s^{\frac{1}{2}}u}_1 , \mathbf{D_s^{\frac{1}{2}}u}_2, \ldots, \mathbf{D_s^{\frac{1}{2}}u}_{n_s}$.

From Theorem 8.1.1 of \cite{golub2012matrix}, there exists an orthogonal matrix $\mathbf{Q} = [\mathbf{q}_1,\mathbf{q}_2,...,\mathbf{q}_{n_s}]$ such that
\[
\mathbf{Q^TL_{sym}Q} = diag(\lambda_1,\lambda_2,...,\lambda_{n_s}).
\]
It shows that $\mathbf{q}_1,\mathbf{q}_2,...,\mathbf{q}_{n_s}$ are linearly independent eigenvectors of $\mathbf{L_{sym}}$, so $\mathbf{D_s^{\frac{1}{2}}u}_1 , \mathbf{D_s^{\frac{1}{2}}u}_2, \ldots, \mathbf{D_s^{\frac{1}{2}}u}_{n_s}$ are linearly independent. As $G_s$ is a connected graph, $\mathbf{D_s^{\frac{1}{2}}}$ is invertible, and it is obvious that $\mathbf{u}_1 , \mathbf{u}_2, \ldots, \mathbf{u}_{n_s}$ are linearly independent.

Additionally, as $\mathbf{L_{rw}}\frac{\mathbf{e}}{\Vert \mathbf{e} \Vert_2} = \mathbf{0}$, we have $1 - \sigma_1 = \lambda_1 = 0$ and $\mathbf{u}_1 = \frac{\mathbf{e}}{\Vert \mathbf{e} \Vert_2}$, so $\sigma_1 = 1$.

As $G_s$ is a connected and non-bipartite graph, by Lemma 1.7 of \cite{chung1997spectral}, we have $1-\sigma_2 = \lambda_2 > 0$ and $1-\sigma_{n_s} = \lambda_{n_s} < 2$. Therefore, $\sigma_2 < 1$ and $\sigma_{n_s} > -1$.
\end{proof}

\begin{theorem}
Let $G_s = (V_s,E_s)$ be a connected and non-bipartite graph with $n_s$ nodes and $m_s$ edges, when $k\rightarrow\infty$, $\mathbf{p}_k$ defined by Eq. (\ref{eq:PonNrw}) converges to $\alpha_1\mathbf{u}_1$ where $\alpha_1$ is the nonzero weighting portion of $\mathbf{p}_0$ on the eigenvector $\mathbf{u}_1$.
\label{theoremone}
\end{theorem}

\begin{proof}
By Lemma \ref{lemmaone}, $\mathbf{u}_1 , \mathbf{u}_2, \ldots, \mathbf{u}_{n_s}$ are linearly independent normalized eigenvectors of $\mathbf{N_{rw}}$, so there exist $\alpha_1,\alpha_2,...,\alpha_{n_s}$ such that $\mathbf{p}_0 = \sum_{i=1}^{n_s}\alpha_i\mathbf{u}_i$.

As $\mathbf{u}_1 = \frac{\mathbf{e}}{\Vert \mathbf{e} \Vert_2}$ and $\mathbf{p}_0$ is the initial probability density evenly assigned on the seeds, $\mathbf{u}_1$ and $\mathbf{p}_0$ are not orthogonal. It shows that $\alpha_1$ is the nonzero weighting portion of $\mathbf{p}_0$ on the eigenvector $\mathbf{u}_1$. Then
\begin{equation*}
\mathbf{p}_k = \mathbf{(N_{rw})}^k\mathbf{p}_0 = \sum_{i = 1}^{n_s}\alpha_i \sigma_i^k\mathbf{u}_i
 = \sigma_1^k\sum_{i = 1}^{n_s}\alpha_i(\frac{\sigma_i}{\sigma_1})^k\mathbf{u}_i.
\end{equation*}
Since $1 = \sigma_1 > \sigma_2 \geq ... \geq \sigma_{n_s} > -1$ and $\alpha_1\neq 0$, for all $i = 2,3,...,n_s$, we have
\[
\lim_{k\rightarrow\infty}(\frac{\sigma_i}{\sigma_1})^k = 0,
\]
and
\[
\lim_{k\rightarrow\infty}\mathbf{p}_k = \lim_{k\rightarrow\infty}\alpha_1\sigma_1^k\mathbf{u}_1 = \alpha_1\mathbf{u}_1.
\]
\end{proof}

Obviously, Lemma \ref{lemmaone} and Theorem \ref{theoremone} also hold for $\mathbf{N_{rw}}$ based on Eq. (\ref{eq:lightlazyRW}) or Eq. (\ref{eq:lazyRW}), which is the transition matrix of modified graph with a weighting loop at each node.

By Theorem \ref{theoremone}, we have the following corollary.
\begin{corollary}
Suppose matrix $\mathbf{N_{rw}}$ defined by Eq. (\ref{eq:per}) has $n_s$ eigenvalues $\mu_1, \mu_2, ... , \mu_{n_s}$ with an associated collection of linearly independent eigenvectors $\{\mathbf{v}_1, \mathbf{v}_2, ... , \mathbf{v}_{n_s}\}$. Moreover, we assume that $|\mu_1| > |\mu_2| \geq ... \geq |\mu_{n_s}| $. Then we have
\[
\lim_{k\rightarrow\infty}\mathbf{p}_k = \lim_{k\rightarrow\infty}\beta_1\mu_1^k\mathbf{v}_1,
\]
where $\mathbf{p}_k$ defined by Eq. (\ref{eq:PonNrw}) and $\beta_1$ is the weighting portion of $\mathbf{p}_0$ on the eigenvector $\mathbf{v}_1$.
\label{corone}
\end{corollary}

Corollary \ref{corone} indicates that $\mathbf{p}_k$ converges to $\beta_1\mathbf{v}_1$ if $\mu_1 = 1$ and $\beta_1 \neq 0$.

Below, we provide some discussion on the convergence of $\mathbf{p}_k$ defined by Eq. (\ref{eq:PonNrwT}). Firstly, we give the following theorem.
\begin{theorem}
Every real square matrix $\mathbf{X}$ is a product of two real symmetric matrices, $\mathbf{X}=\mathbf{YZ}$ where $\mathbf{Y}$ is invertible.
\label{matfac}
\end{theorem}

We will not include the proof of Theorem \ref{matfac} here. The interested reader is referred to \cite{bosch1986factorization}.

By Theorem \ref{matfac}, we have $\mathbf{N_{rw}}=\mathbf{PU}$ where $\mathbf{P}$ and $\mathbf{U}$ are symmetric matrices, and $\mathbf{P}$ is invertible. Then,
\begin{equation}
\mathbf{N_{rw}}=\mathbf{PU}=\mathbf{P(UP)P^{-1}}=\mathbf{P(PU)^TP^{-1}}=\mathbf{PN_{rw}^TP^{-1}}.
\label{similarity}
\end{equation}
According to Eq. (\ref{similarity}), we have
\begin{equation}
\mathbf{N_{rw}v} = \lambda \mathbf{v} \quad \Leftrightarrow \quad  \mathbf{N_{rw}^T(P^{-1}v)} = \lambda \mathbf{(P^{-1}v)},
\label{equal}
\end{equation}
where $\bf{v}$ is a nonzero vector. It shows that $\mathbf{N_{rw}}$  and $\mathbf{N_{rw}^T}$ share the same set of eigenvalues and the corresponding eigenvector of $\mathbf{N_{rw}^T}$ is $\mathbf{P^{-1}v}$ where $\mathbf{v}$ is the eigenvector of $\mathbf{N_{rw}}$.

By Lemma \ref{lemmaone}, we know that $\mathbf{N_{rw}}$ defined by Eq. (\ref{Eq1}) has $n_s$ linearly independent normalized eigenvectors $\mathbf{u}_1 , \mathbf{u}_2, \ldots, \mathbf{u}_{n_s}$. Then by Eq. (\ref{equal}), $\mathbf{N_{rw}^T}$ has $n_s$ linearly independent eigenvectors $\mathbf{P^{-1}u}_1 , \mathbf{P^{-1}u}_2, \ldots, \mathbf{P^{-1}u}_{n_s}$. By Theorem \ref{theoremone}, we have the following theorem.
\begin{theorem}
Let $G_s = (V_s,E_s)$ be a connected and non-bipartite graph with $n_s$ nodes and $m_s$ edges, when $k\rightarrow\infty$, $\mathbf{p}_k$ defined by Eq. (\ref{eq:PonNrwT}) converges to $\gamma_1\mathbf{P^{-1}u}_1$ where $\gamma_1$ is the nonzero weighting portion of $\mathbf{p}_0$ on the eigenvector $\mathbf{P^{-1}u}_1$.
\label{theoremtwo}
\end{theorem}

The proof of Theorem \ref{theoremtwo} is similar to Theorem \ref{theoremone}, hence we omit the details here.

Obviously, Theorem \ref{theoremtwo} also holds for $\mathbf{N_{rw}}$ based on Eq. (\ref{eq:lightlazyRW}) or Eq. (\ref{eq:lazyRW}), which is the transition matrix of modified graph with a weighting loop at each node.

By Eq. (\ref{equal}) and Theorem \ref{theoremtwo}, we have the following corollary.
\begin{corollary}
Suppose matrix $\mathbf{N_{rw}}$ defined by Eq. (\ref{eq:per}) has $n_s$ eigenvalues $\mu_1, \mu_2, ... , \mu_{n_s}$ with an associated collection of linearly independent eigenvectors $\{\mathbf{v}_1, \mathbf{v}_2, ... , \mathbf{v}_{n_s}\}$. Moreover, we assume that $|\mu_1| > |\mu_2| \geq ... \geq |\mu_{n_s}| $. Then we have
\[
\lim_{k\rightarrow\infty}\mathbf{p}_k = \lim_{k\rightarrow\infty}\delta_1\mu_1^k\mathbf{P^{-1}v}_1,
\]
where $\mathbf{p}_k$ defined by Eq. (\ref{eq:PonNrwT}) and $\delta_1$ is the weighting portion of $\mathbf{p}_0$ on the eigenvector $\mathbf{P^{-1}v}_1$ of $\mathbf{N_{rw}^T}$.
\label{cortwo}
\end{corollary}

Corollary \ref{cortwo} indicates that $\mathbf{p}_k$ converges to $\delta_1\mathbf{P^{-1}v}_1$ if $\mu_1 = 1$ and $\delta_1 \neq 0$.

When $k\rightarrow\infty$, Theorem \ref{theoremone} and Corollary \ref{corone} state that the local spectral subspace $\mathcal{V}_d^{(k)}$ built on ``inverse random walk'' approaches the eigenspace associated with eigenvector of $\mathbf{N_{rw}}$ with the largest eigenvalue, and Theorem \ref{theoremtwo} and Corollary \ref{cortwo} indicate that the local spectral subspace $\mathcal{V}_d^{(k)}$ built on ``regular random walk'' approaches the eigenspace associated with eigenvector of $\mathbf{N_{rw}^T}$ with the largest eigenvalue. Our interest now, though, is not in the limiting case when $k$ is large, but for a much more modest number of diffusion steps to reveal the local property around the seeds. Based on different local spectral diffusions in Eq. (\ref{Eq1}) - Eq. (\ref{eq:per}), we have a set of local spectral subspace definitions. Based on Eq. (\ref{eq:PonNrwT}) and Eq. (\ref{eq:PonNrw}), we have two sets of local spectral subspace definitions.

Experiments in Section \ref{sec:experiments} show that the definitions on $\mathbf{N_{rw}^T}$, which is on the regular random walk, is considerably better than that on $\mathbf{N_{rw}}$ for the detection accuracy. However, the definitions on $\mathbf{N_{rw}}$ corresponding to the inverse random walk also show high accuracy as compared with the baselines. Our results show that the local approximation built on ``regular random walk'' shows higher accuracy than that built on ``inverse random walk''.

\subsection{Local Community Detection}
We modify the optimization problem shown in Eq. (\ref{eq:oneNorm}) by relaxing each element in the indicator vector ${\bf y}$ to be nonnegative,
and approximating the global spectral subspace by the local spectral subspace.
Thus, we seek a relaxed sparse vector in the local spectral subspace by solving a linear programming problem.

\begin{equation}
\begin{aligned}
&\min ~~ \Vert\bf{y}\Vert_1 = \mathbf{e^Ty} \\
\emph{s.t.}& ~(1)~ \mathbf{y = V}_d^{(k)}\mathbf{u}, \\
           & ~(2)~ {\mathbf y }\geq 0, \\
           & ~(3)~y_i \geq \frac{1}{|S|},~i \in S.
\end{aligned}
\vspace{3mm}
\end{equation}

This is an $\ell_{1}$ norm approximation for finding a sparse linear coding that indicates a small community containing the seeds with $\bf{y}$ in the local spectral subspace spanned by the column vectors of $\mathbf{V}_d^{(k)}$.
The $i$th entry $y_i$ indicates the likelihood of node $i$ belonging to the target community.

1) If $|T|$ is known, we then sort the values in $\bf{y}$ in non-ascending order and select the corresponding $|T|$ nodes with the higher belonging likelihood as the output community.

2) If $|T|$ is unknown, we use a heuristic
to determine the community boundary.
We sort the nodes based on the element values of $\bf{y}$ in the decreasing order, and find a set $S_{k^*}$ with the first $k^*$ nodes having a comparatively low conductance.
Specifically, we start from an index $k_0$ where set $S_{k_0}$ contains all the seeds.
We then generate a sweep curve $\Phi(S_{k})$ by increasing index $k$.
Let $k^*$ be the value of $k$ where $\Phi(S_{k})$ achieves a first local minimum.
The set $S_{k^*}$ is regarded as the detected community.

We determine a local minima as follows.
If at some point $k^*$ when we are increasing $k$, $\Phi(S_{k})$ stops decreasing, then this $k^*$ is a candidate point for the local minimum.
If $\Phi(S_{k})$ keeps increasing after $k^*$ and it eventually becomes higher than $\beta \Phi(S_{k^*})$, then we take $k^*$ as a valid local minimum.
We experimented with several values of $\beta$ on a small trial of data and found that $\beta = 1.02$ gives good performance across all the datasets.

Denote the corresponding local community detection methods based on Eq. (\ref{Eq1}) - Eq. (\ref{eq:per}) as: LRw (LOSP based on Standard Random Walk), LLi (LOSP based on Light Lazy Random Walk), LLa (LOSP based on Lazy Random Walk) and LPr (LOSP based on PPR). 

\section{Experiments And Results}
\label{sec:experiments}
We implement the family of local spectral methods (LOSPs) in Matlab\footnote{https://github.com/PanShi2016/LOSP\_Plus}$^{,}$\footnote{https://github.com/JHL-HUST/LOSP\_Plus}
and thoroughly compare them with state-of-the-art localized community detection algorithms on the 28 LFR datasets as well as the 8 real-world networks across multiple domains.
For the 5 SNAP datasets, we randomly locate 500 labeled ground truth communities on each dataset, and randomly pick three exemplary seeds from each target community.
For the 28 LFR datasets and the 3 Biology datasets, we deal with every ground truth community and randomly pick three exemplary seeds from each ground truth community.
We pre-process all real-world datasets by sampling, and apply the local spectral methods for each network.

\subsection{Statistics on Sampling}
In order to evaluate the effectiveness of the sampling method in Algorithm \ref{alg:one}, we empirically set $(N_1,N_2,k) = (300,5000,3)$ to control the subgraph size and experiment with the upper bound of BFS rounds $t$ from 1 to 5 to extract different sampling rate on Amazon, as shown in Table \ref{sampleTest}. For notations, the coverage indicates the average fraction of ground truth covered by the sampled subgraph, and $n_s/n$ indicates the sampling rate which is the average fraction of subgraph size as compared with the original network scale.

Table \ref{sampleTest} shows that there is a $7.2\%$ significant improvement on coverage when we increase the upper bound of BFS rounds $t$ from 1 to 2, but there is only $0.4\%$ improvement when $t$ continue increases from 2 to 5. On the other hand, the sampling rate is only $0.1\%$ and the sampling procedure is very fast in 0.730 seconds for $t = 2$. For these reasons, we set $(N_1,N_2,t,k) = (300,5000,2,3)$ to trade off among the coverage, sampling rate and running time for the sampling method in our experiments.

Table \ref{realsample} provides statistics on real-world networks for the sampling method in Algorithm \ref{alg:one}.
For SNAP datasets, our sampling method has a high coverage with reasonable sample size, covering about $96\%$ ground truth with a small average sampling rate of $0.1\%$, and the sampling procedure is within 14 seconds.
For the Biology networks, which are comparatively denser, the sampled subgraph covers about $91\%$ ground truth with a relatively high sampling rate, and the sampling procedure is very fast in less than 0.2 seconds.

\begin{table}[htbp]
	\renewcommand{\arraystretch}{1.3}
	\caption{Test parameter $t$ (upper bound of BFS rounds) for the sampling on Amazon.}
	\label{sampleTest}
	\centering
	\begin{tabular}{l | c c c c c}
		\hline
		\multirow{2}{*}{\bf{Statistics}}  & \multicolumn{5}{c}{\bf{Upper bound of BFS rounds}}  \\
		&1    &2    &3    &4     &5 \\
		\hline
		\bf{Coverage} & 0.918  & 0.990  & 0.991  & 0.993  & 0.994   \\
		\hline
		$n_s$ & 13  & 34  & 70  & 184  & 312   \\
		\hline
        $n_s/n$ & 0.00004 & 0.00010 & 0.00021 & 0.00055 & 0.00093 \\
        \hline
		\bf{Time (s)} & 0.374  & 0.730  & 2.208  & 4.361  & 7.028 \\
		\hline
	\end{tabular}
\end{table}

\begin{table}[htbp]
\renewcommand{\arraystretch}{1.3}
\caption{Statistics on average values for the sampling on real-world networks.}
\label{realsample}
\centering
\begin{tabular}{ l l | r | r | r | r }
\hline
&\bf{Networks} & \bf{Coverage} & \bf{$n_s$} & \bf{$n_s/n$} & \bf{Time (s)}   \\
\hline
\bf{SNAP} &\bf{Amazon} & 0.990 & 34    & 0.0001  & 0.730 \\
& \bf{DBLP}            & 0.980 & 198   & 0.0006  & 0.720 \\
& \bf{LiveJ}           & 1.000 & 629   & 0.0002  & 19.050 \\
& \bf{YouTube}         & 0.950 & 3237  & 0.0028  & 3.760 \\
& \bf{Orkut}           & 0.870 & 4035  & 0.0013  & 44.430 \\
\hline
&\bf{Average}          & \bf{0.958} & \bf{1627}  & \bf{0.001}  & \bf{13.738} \\
\hline
\bf{Biology} &\bf{DM}  & 0.910 & 2875  & 0.1880  & 0.256 \\
& \bf{HS}              & 0.876 & 2733  & 0.2692  & 0.125 \\
& \bf{SC}              & 0.947 & 3341  & 0.6049  & 0.076 \\
\hline
&\bf{Average}          & \bf{0.911} & \bf{2983}  & \bf{0.3540}  & \bf{0.152} \\
\hline
\end{tabular}
\end{table}

\subsection{Parameter Setup}
To remove the impact of different local spectral methods in finding a local minimum for the community boundary, we use the ground truth size as a budget for parameter testing on the family of LOSP methods.
When we say LOSP, we mean the family of LOSP defined on the $\mathbf{N_{rw}^T}$ Krylov subspace, which is the normal case for random walk diffusion.
A comparison in subsection \ref{sec:LOSPEvaluation} will show that in general, LOSP defined on $\mathbf{N_{rw}^T}$ Krylov subspace outperforms that on $\mathbf{N_{rw}}$ Krylov subspace with respect to the accurate detection.

\textbf{Dimension of the subspace and diffusion steps.}
For local spectral subspace, we need to choose some modest numbers for the step $k$ of random walks and the subspace dimension $d$ such that the probability diffusion does not reach the global stationary. We did a small trial parameter study on all datasets, and found that $d = 2$ and $k = 2$ perform the best in general.

\textbf{Parameters for the random walk diffusion.}
We thoroughly evaluate different spectral diffusion methods on all datasets, as shown in Fig. \ref{fig:RealRW} and Fig. \ref{fig:LFRRW}.
The three columns correspond to light lazy random walk, lazy random walk and personalized pagerank with different $\alpha$ parameters. All three variants degenerate to the standard random walk when $\alpha = 0$.
The results show that light lazy random walk, lazy random walk and personalized pagerank are robust for different $\alpha$ parameters.
The personalized pagerank declines significantly when $\alpha = 1$ as all probability returns to the original seed set.

During the probability diffusion, light lazy random walk and lazy random walk always retain a ratio of probability on the current set of nodes to keep the detected structure to be ``local''. The personalized pagerank always returns a ratio of probability from the current set of nodes to the seed set.
Instead of retaining some probability distribution on the current set of nodes, the personalized pagerank ``shrinks'' some probability to the original seed set. Such process also wants to keep the probability distribution ``local'' but it is not continuous as compared with the previous two methods.

In the following discussion, we set $\alpha = 1$ for light lazy random walk and lazy random walk, and set $\alpha = 0.1$ for personalized pagerank.

\begin{figure*}[!t]
\centering
\subfigure{
 \includegraphics[height=1.6in,width=1.5in]{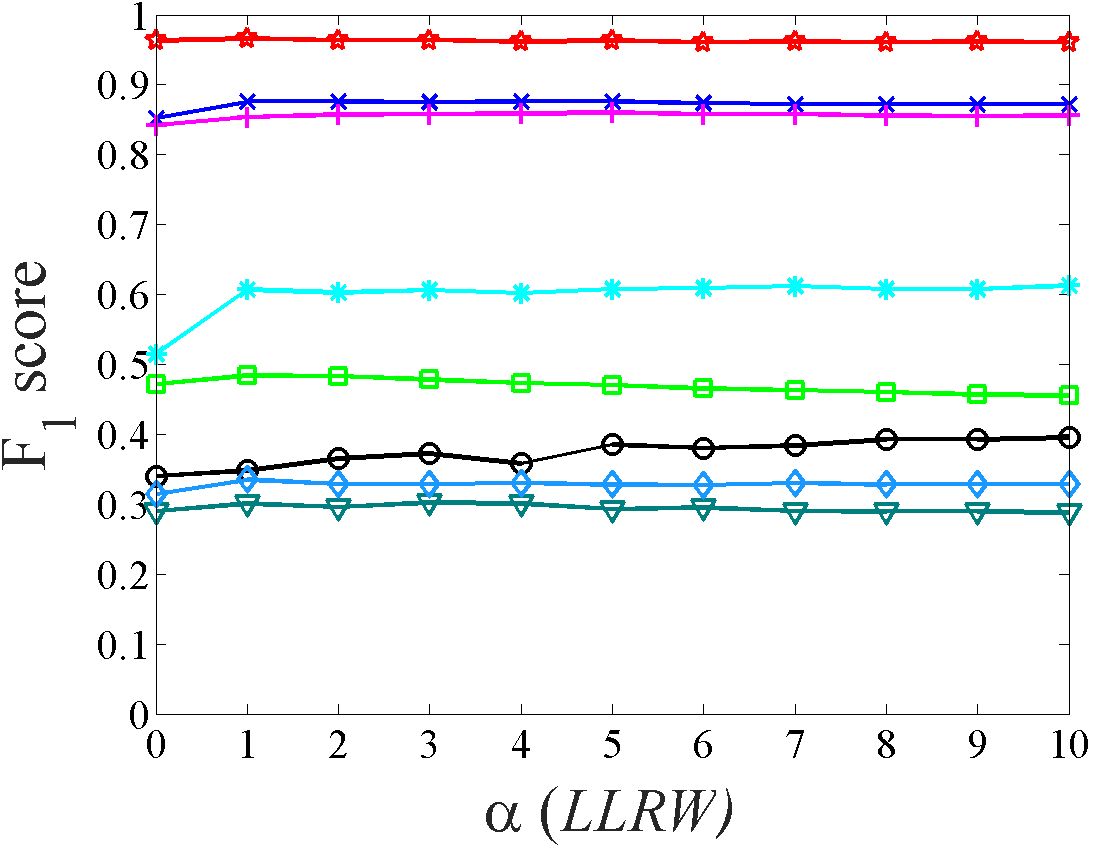}
 \includegraphics[height=1.6in,width=1.45in]{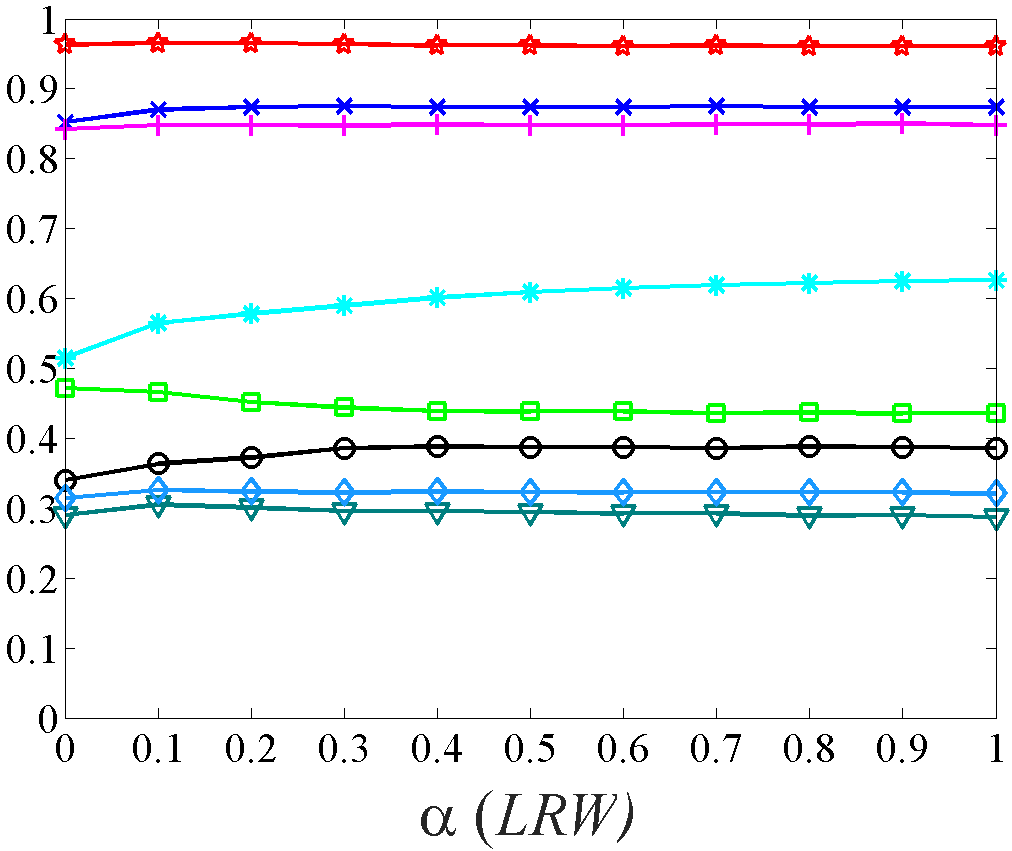}
 \includegraphics[height=1.6in,width=1.9in]{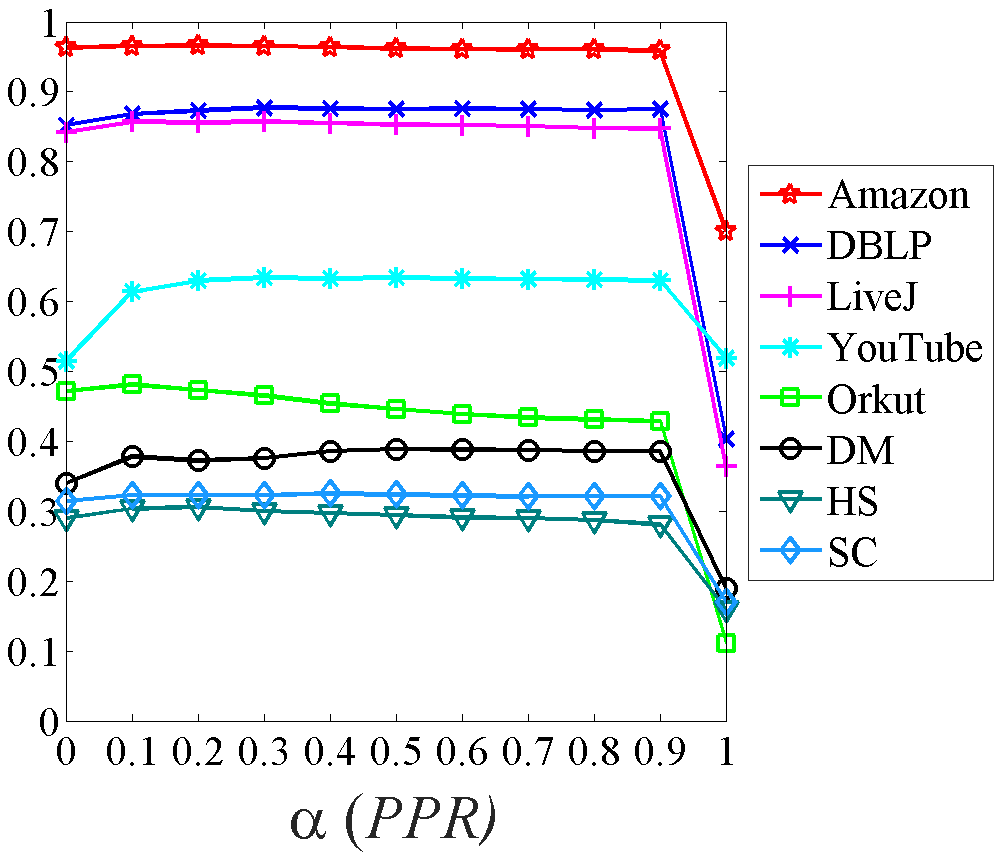}}
 \caption{Evaluation of different diffusion parameters $\alpha$ on real-world datasets (Defined on $\mathbf{N_{rw}^T}$ Krylov subspace, community size truncated by truth size).
 	The three diffusions are robust for different $\alpha$ parameters, except for $\alpha = 1$ on personalized pagerank,  in which case all probability returns to the original seed set.}
 \label{fig:RealRW}
 \end{figure*}

 \begin{figure*}[!t]
 \centering
 \subfigure{
 \includegraphics[height=1.6in,width=1.45in]{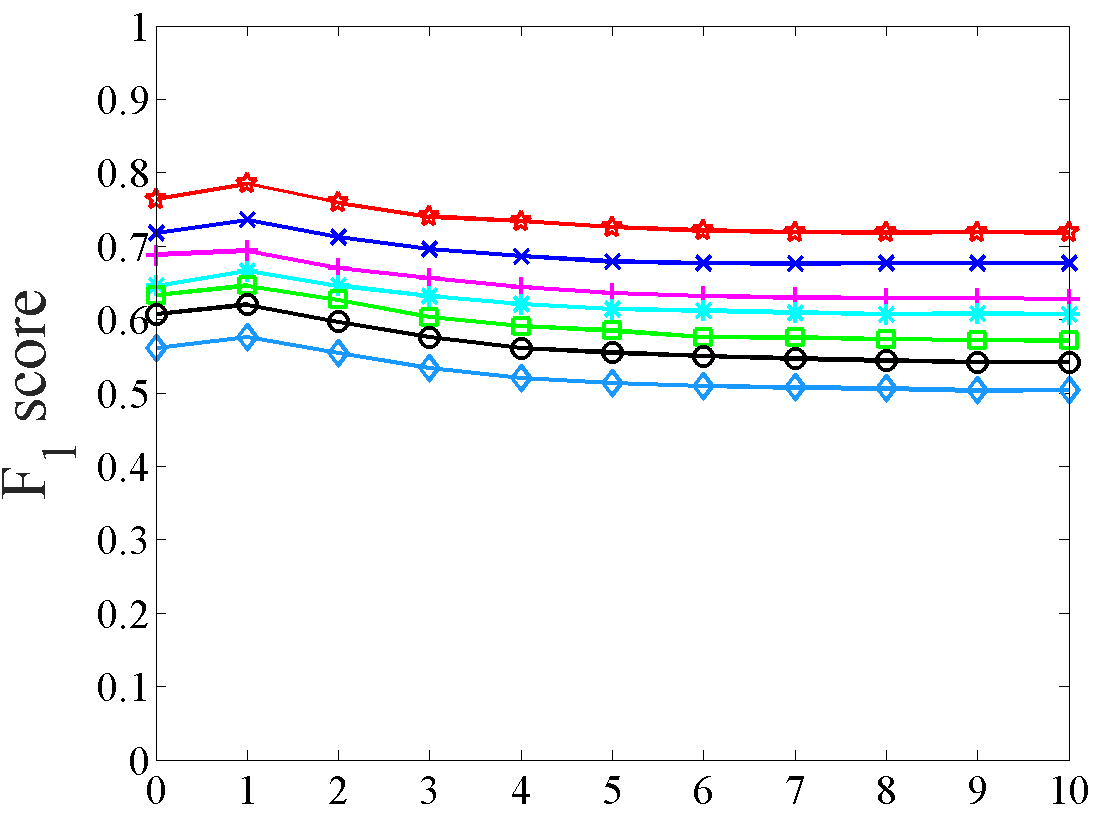}
 \includegraphics[height=1.6in,width=1.45in]{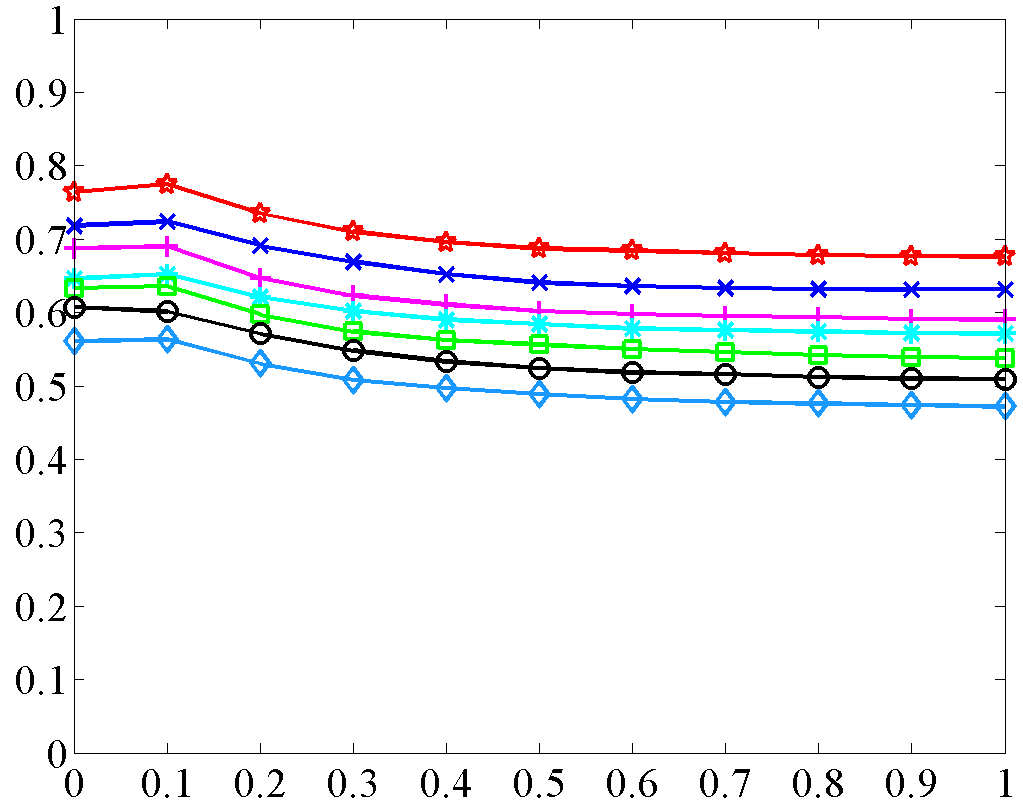}
 \includegraphics[height=1.6in,width=2.1in]{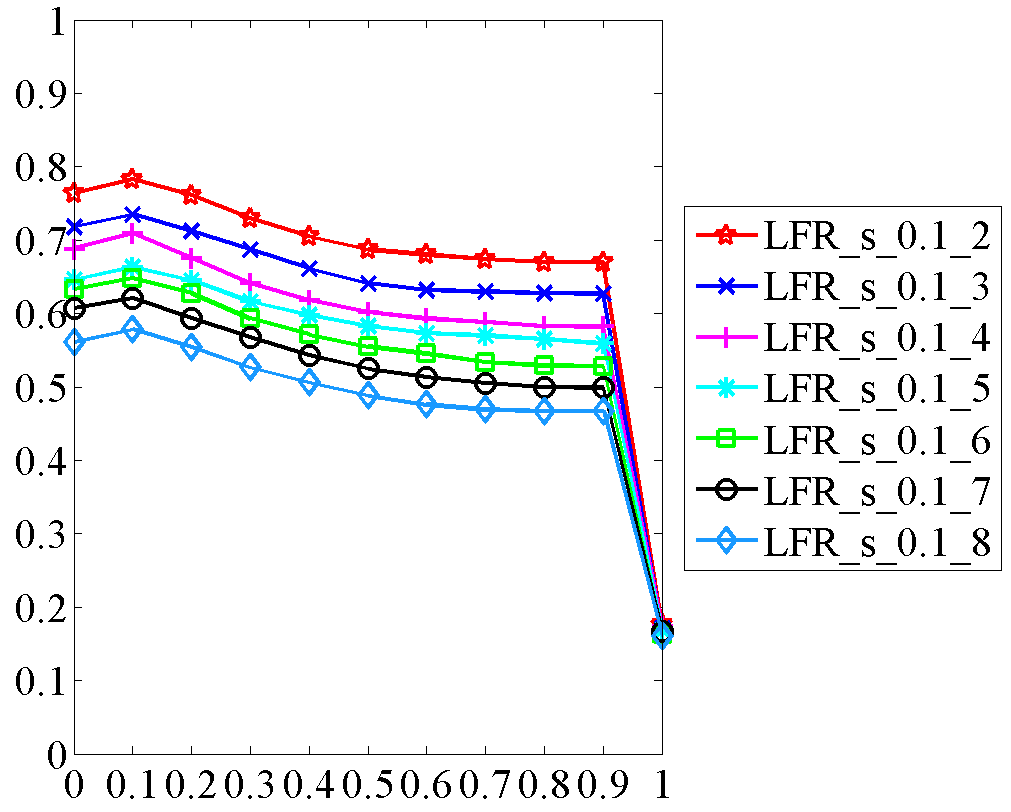}}
 \subfigure{
 \includegraphics[height=1.6in,width=1.45in]{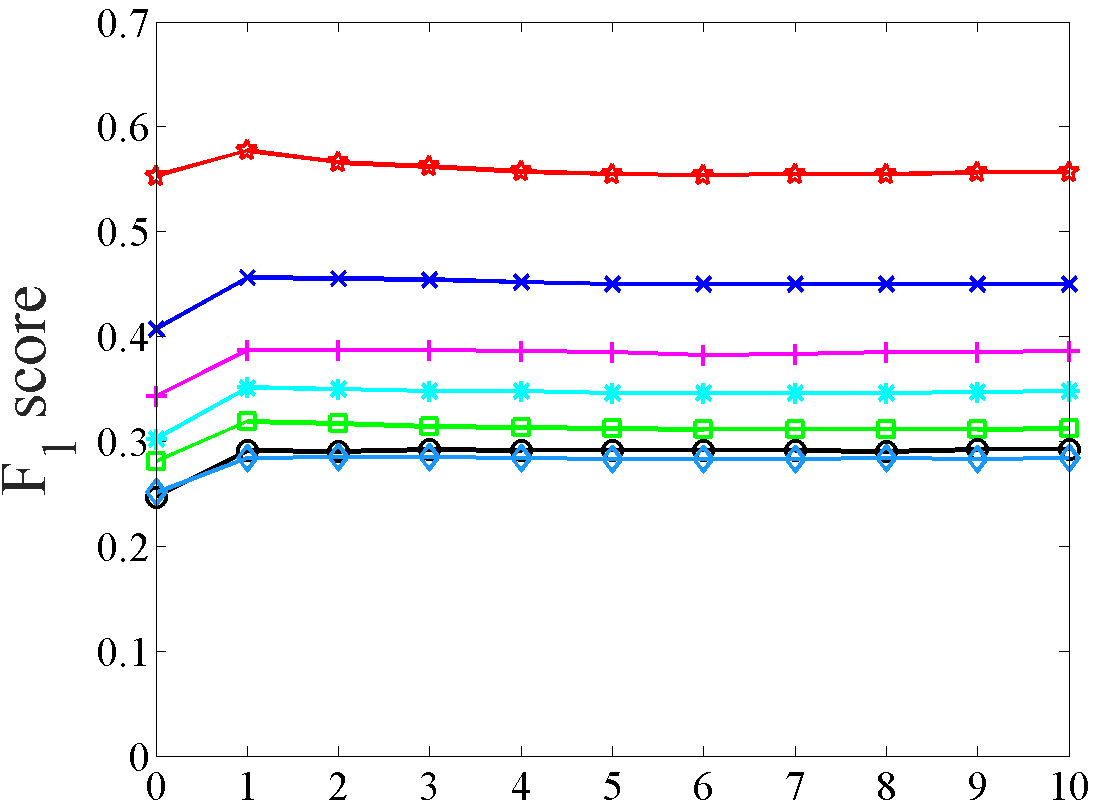}
 \includegraphics[height=1.6in,width=1.45in]{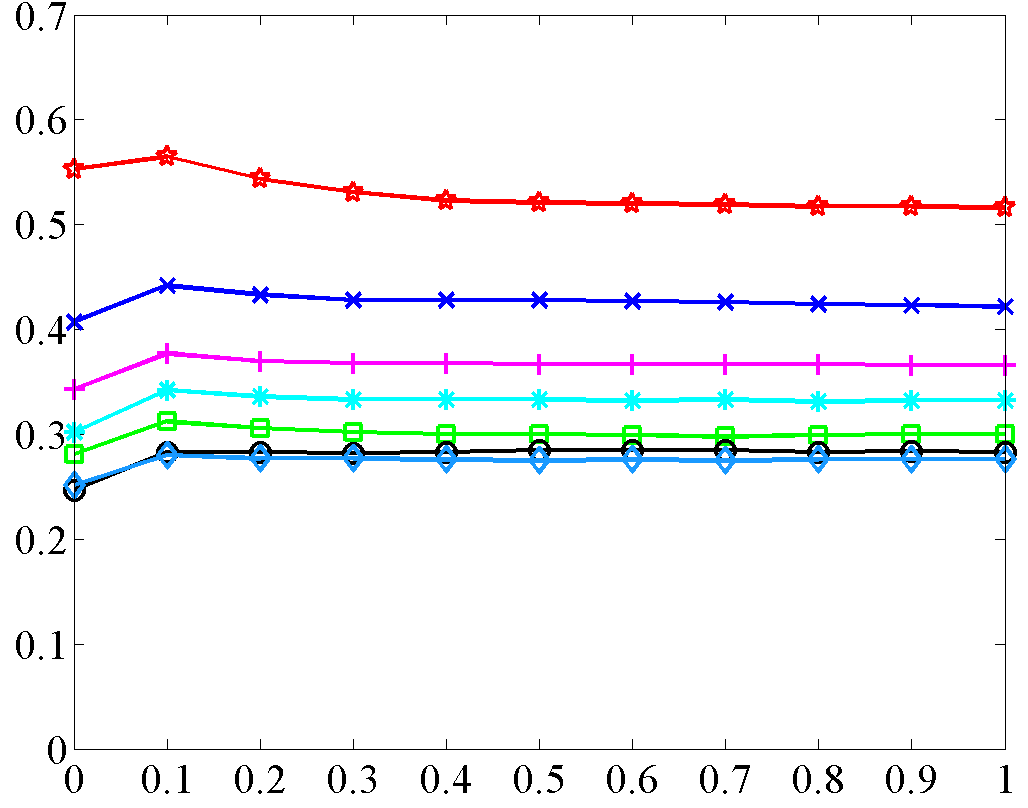}
 \includegraphics[height=1.6in,width=2.1in]{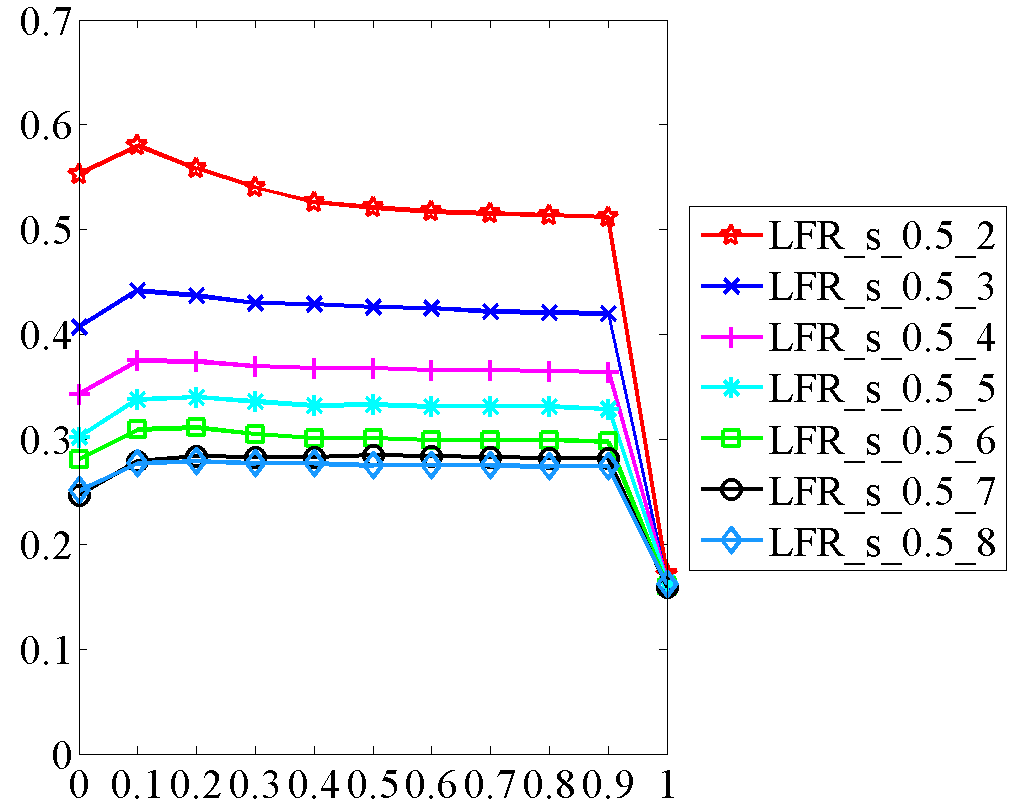}}
 \subfigure{
 \includegraphics[height=1.6in,width=1.45in]{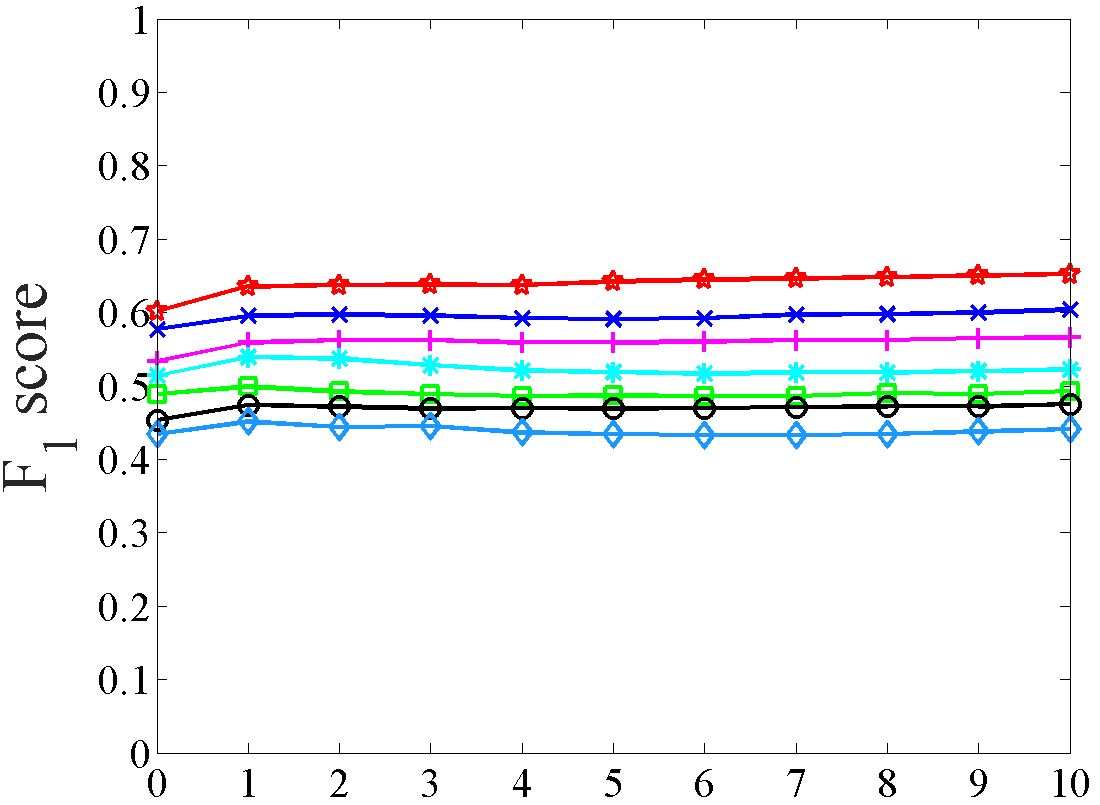}
 \includegraphics[height=1.6in,width=1.45in]{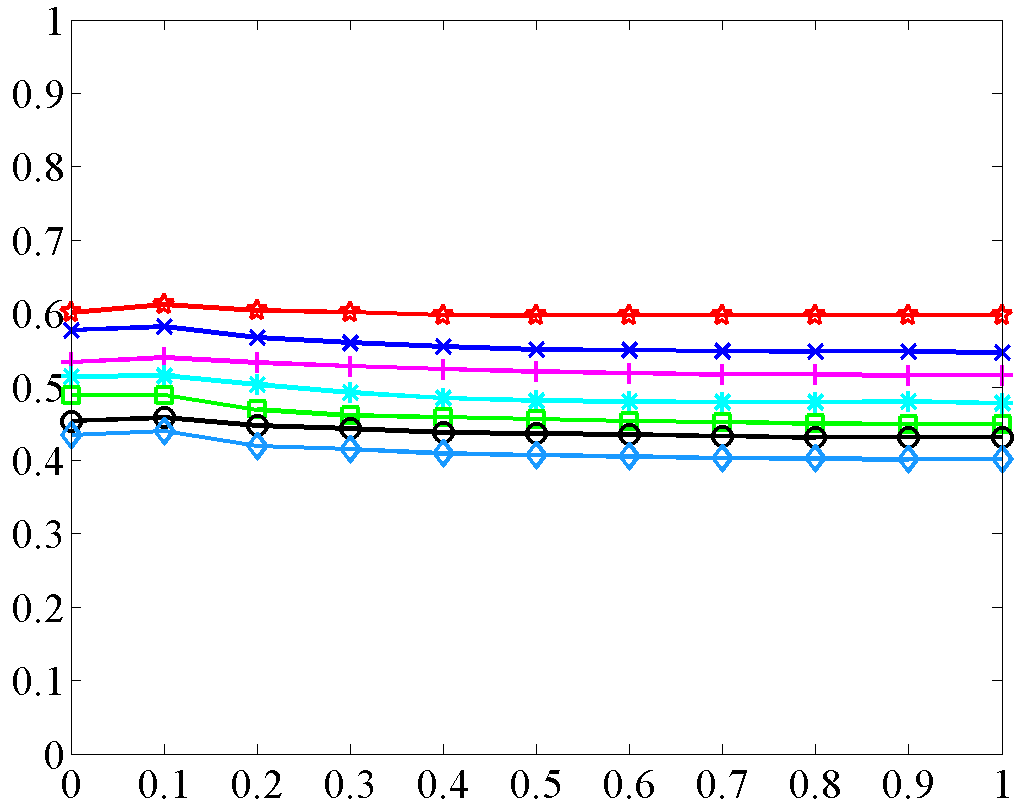}
 \includegraphics[height=1.6in,width=2.1in]{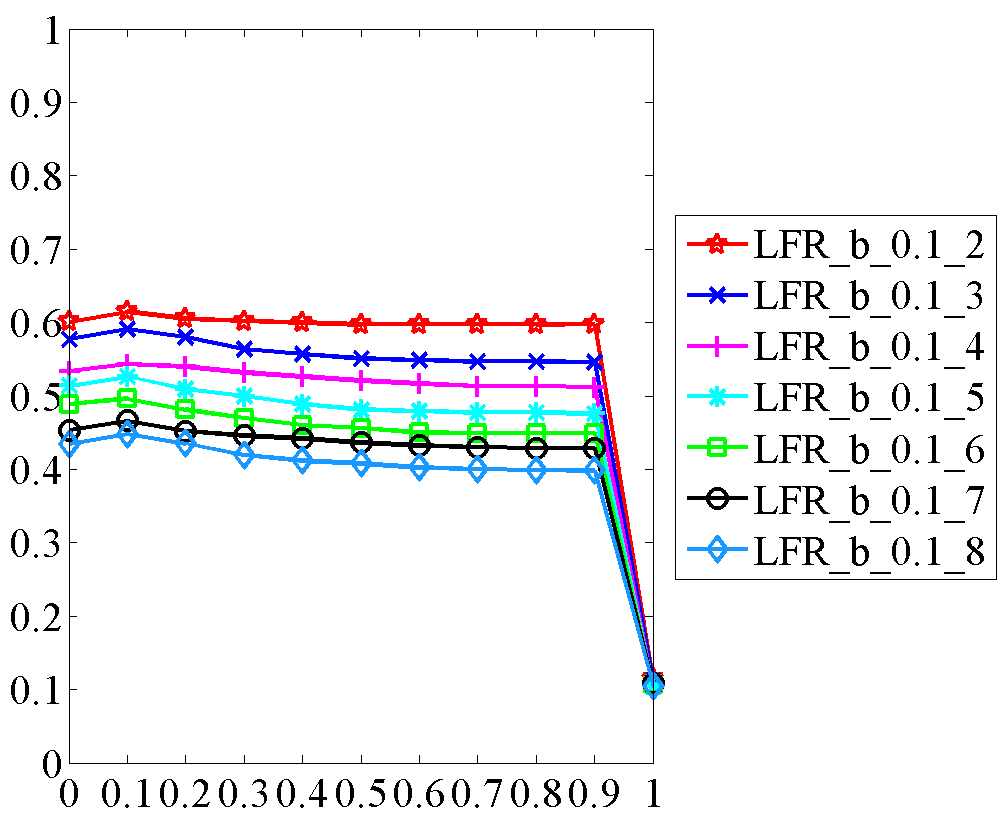}}
 \subfigure{
 \includegraphics[height=1.6in,width=1.45in]{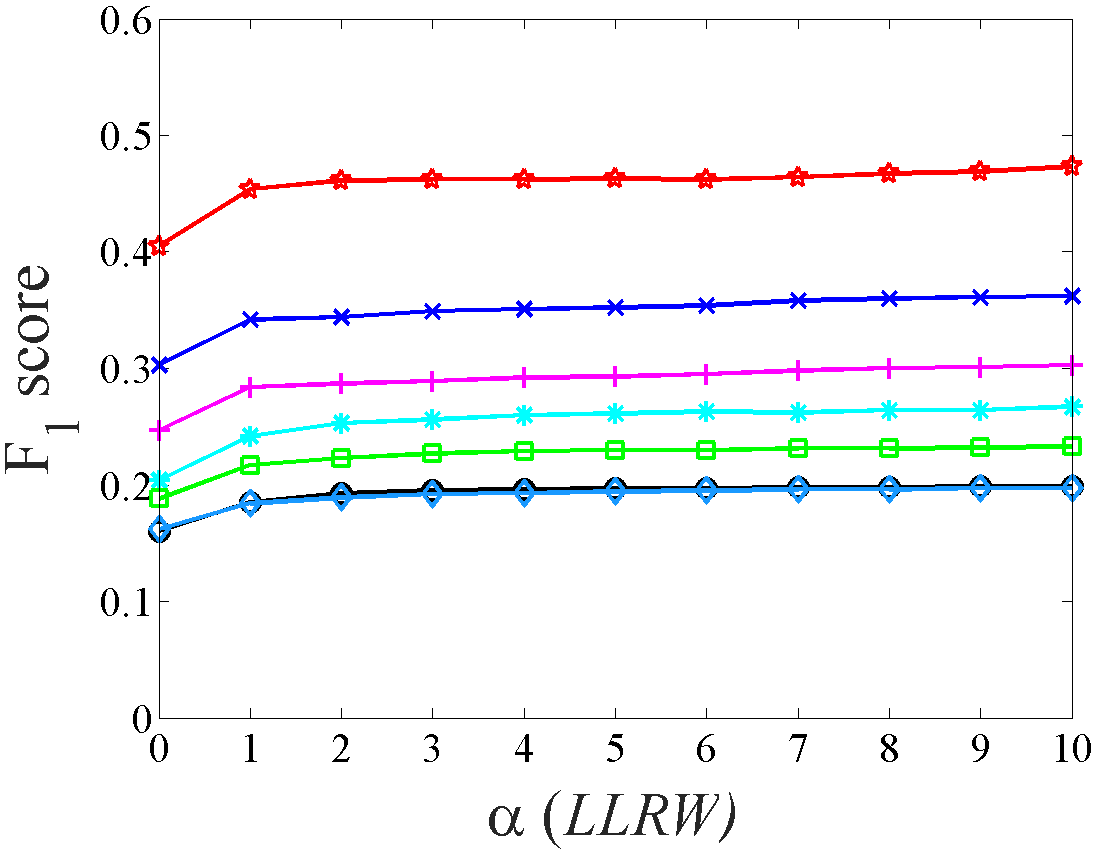}
 \includegraphics[height=1.6in,width=1.45in]{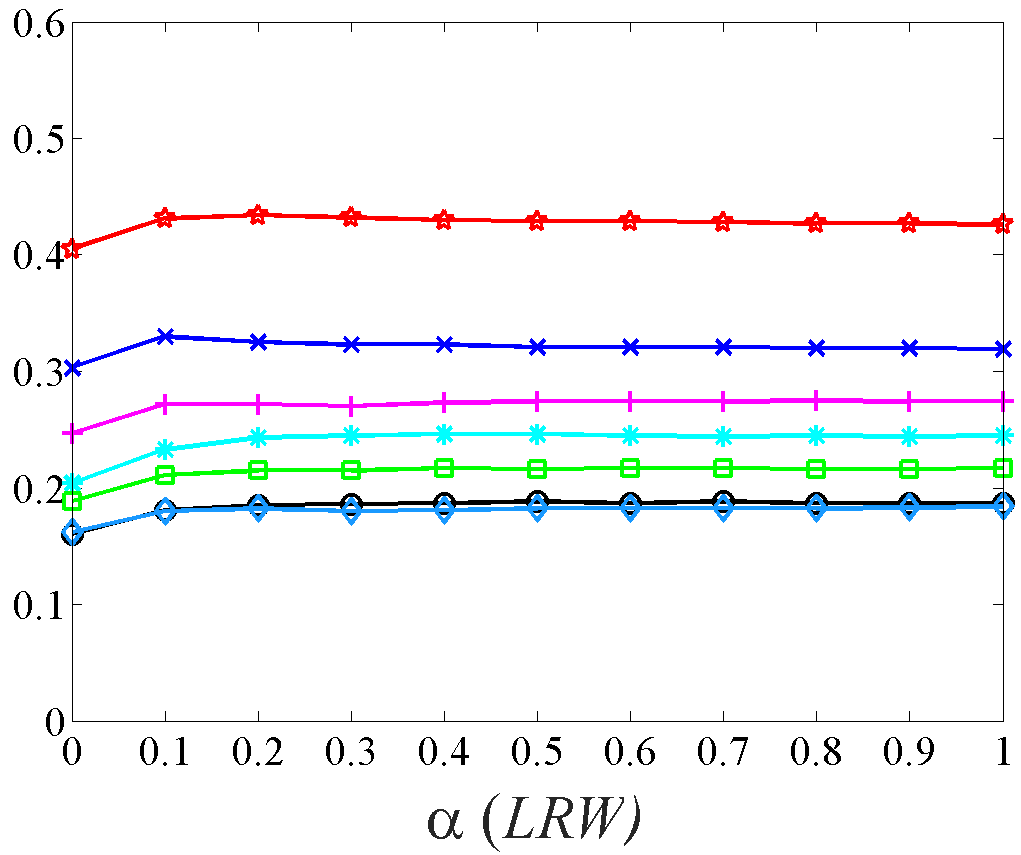}
 \includegraphics[height=1.6in,width=2.1in]{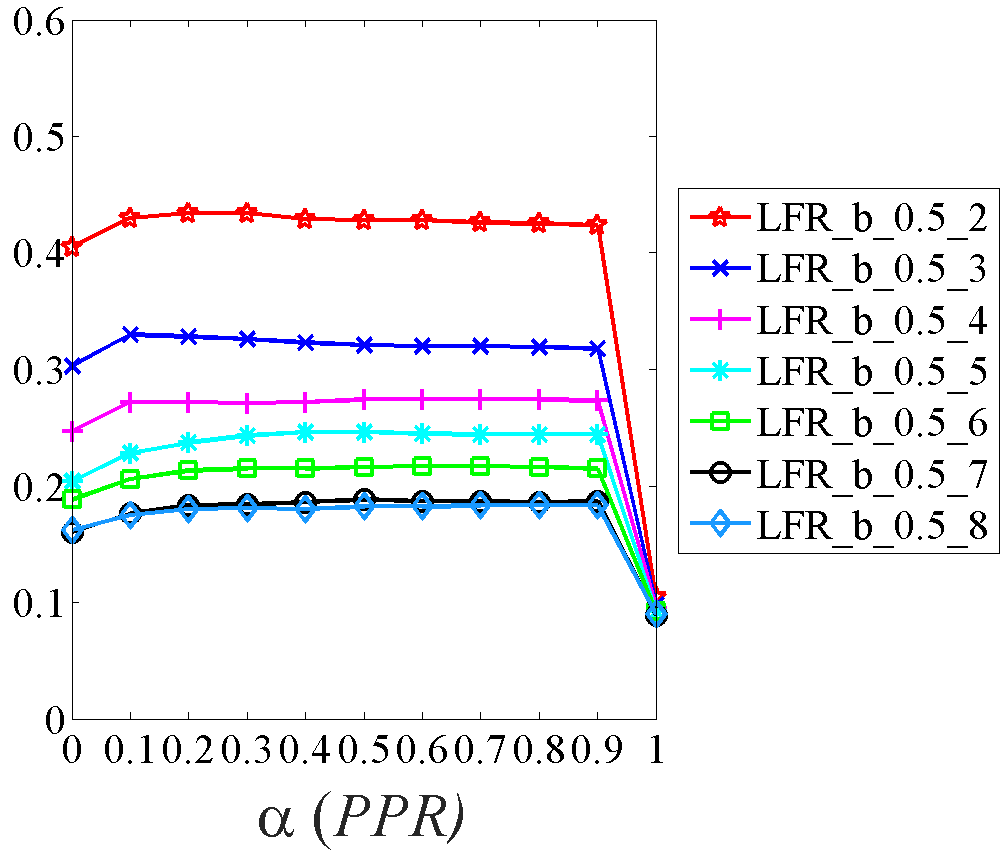}}
\caption{Evaluation of different diffusion parameters $\alpha$ on LFR datasets (Defined on $\mathbf{N_{rw}^T}$ Krylov subspace, community size truncated by truth size).
The three diffusions are robust for different $\alpha$ parameters, except for $\alpha = 1$ on personalized pagerank, in which case all probability returns to the original seed set.}
\label{fig:LFRRW}
\end{figure*}

\begin{figure}[!t]
	\centering
	\includegraphics[width=2.5in]{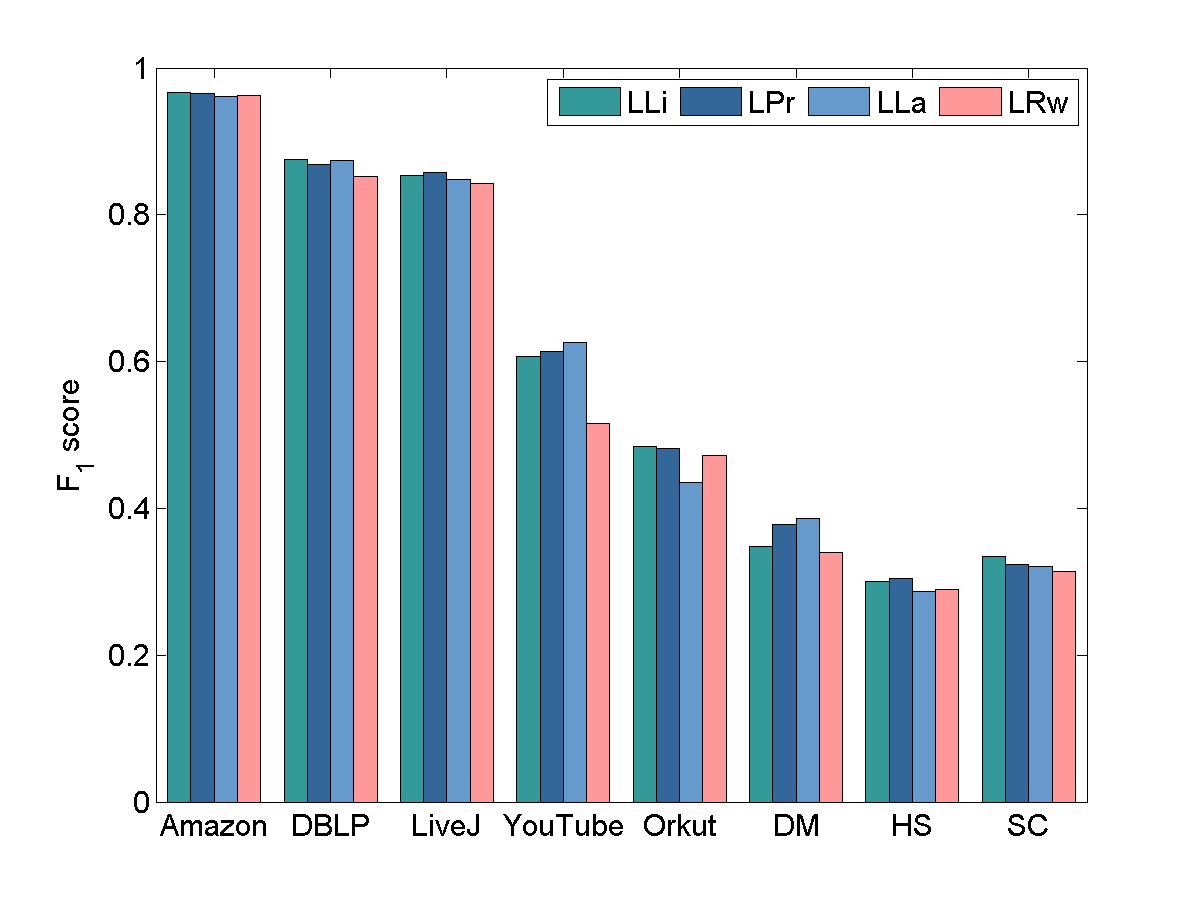}
	\vspace{-1em}
	\caption{Accuracy evaluation of LOSP on real-world networks (Defined on $\mathbf{N_{rw}^T}$ Krylov subspace, community size truncated by truth size).
		LOSPs based on the four diffusions show similar accuracy over all datasets. LLi, LOSP with light lazy, demonstrates slightly higher accuracy on half of the datasets.
	}
    \vspace{1.5em}
	\label{Fig:realCompare}
\end{figure}

\begin{figure}[!t]
	\vspace{-1em}
	\centering
	\subfigure[\bf{LFR\_s\_0.1}]{
		\includegraphics[width=2.5in]{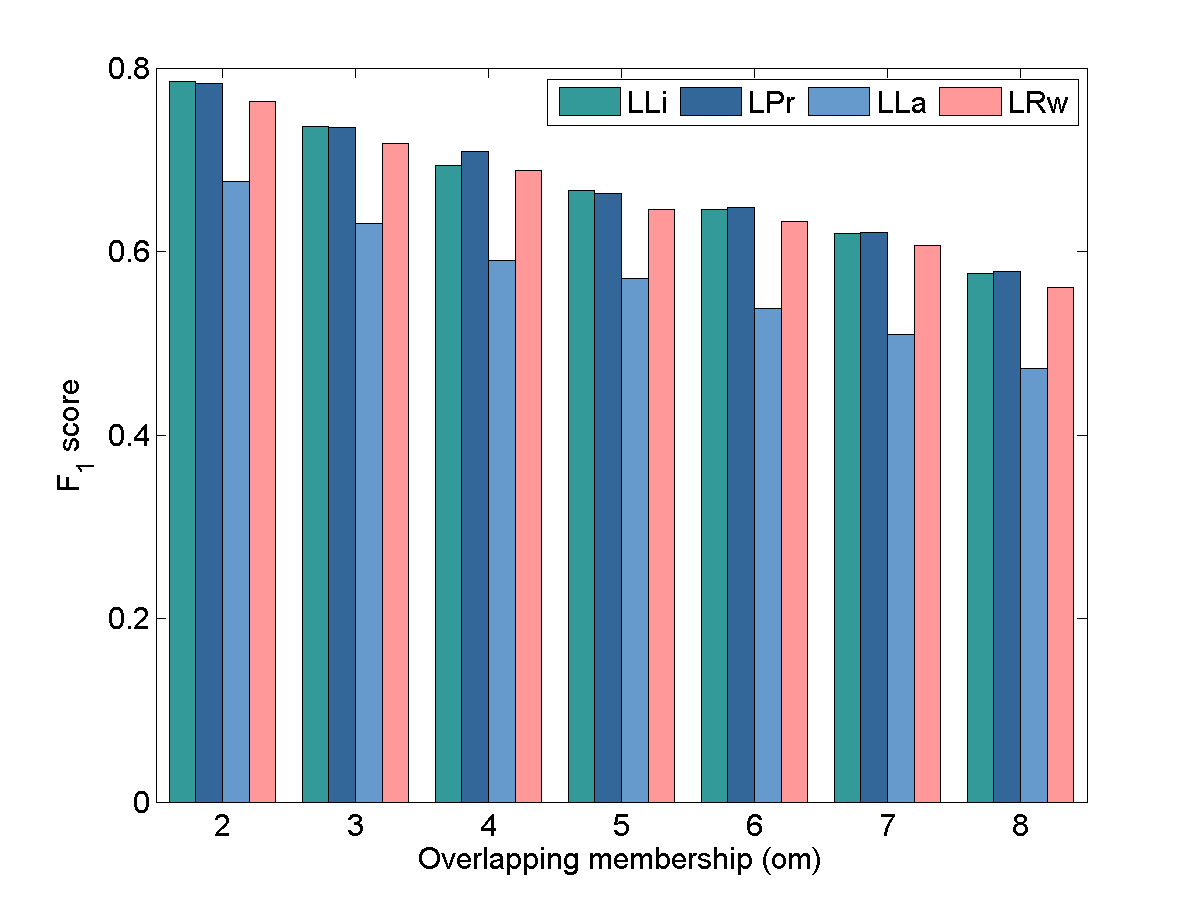}}
		\vspace{-1em}
	\subfigure[\bf{LFR\_s\_0.5}]{
		\includegraphics[width=2.5in]{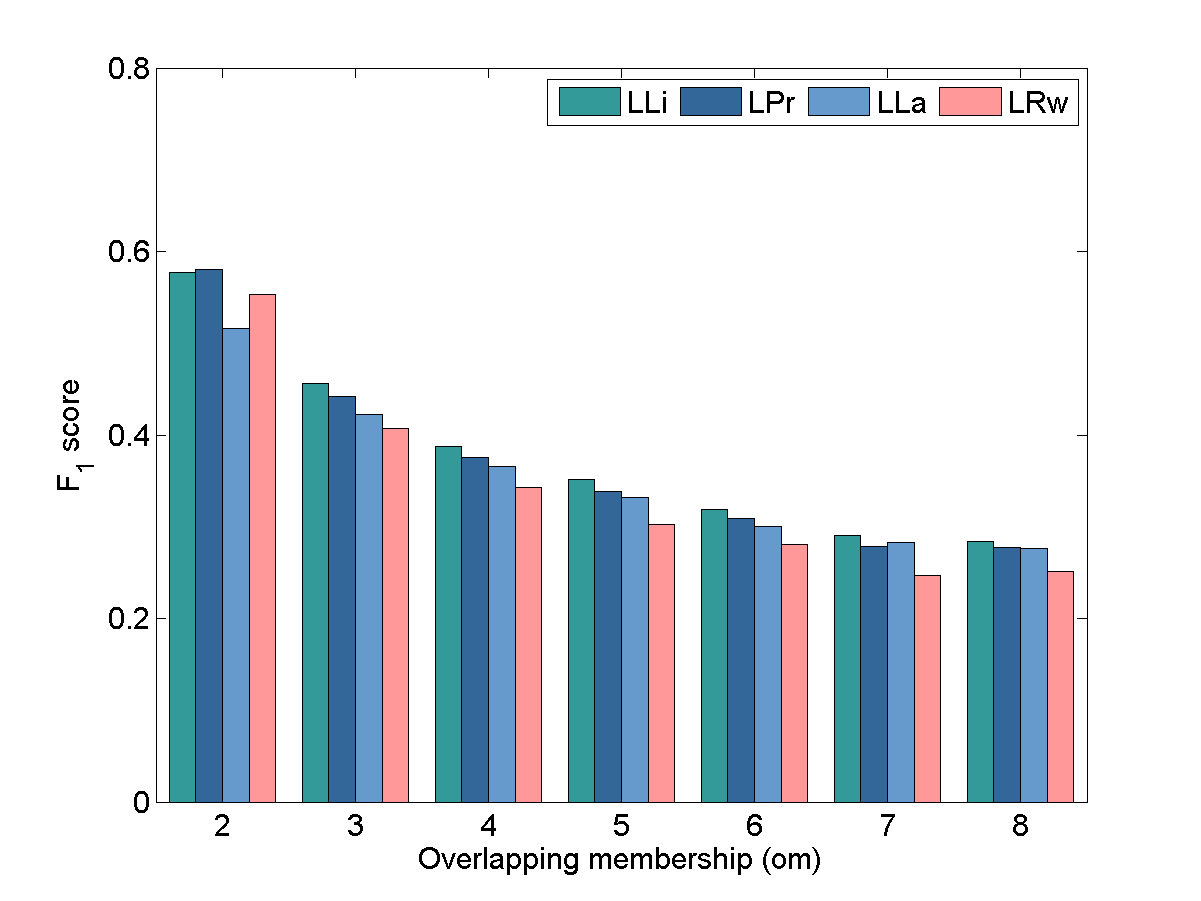}}
	\vspace{-1em}
	\subfigure[\bf{LFR\_b\_0.1}]{
		\includegraphics[width=2.5in]{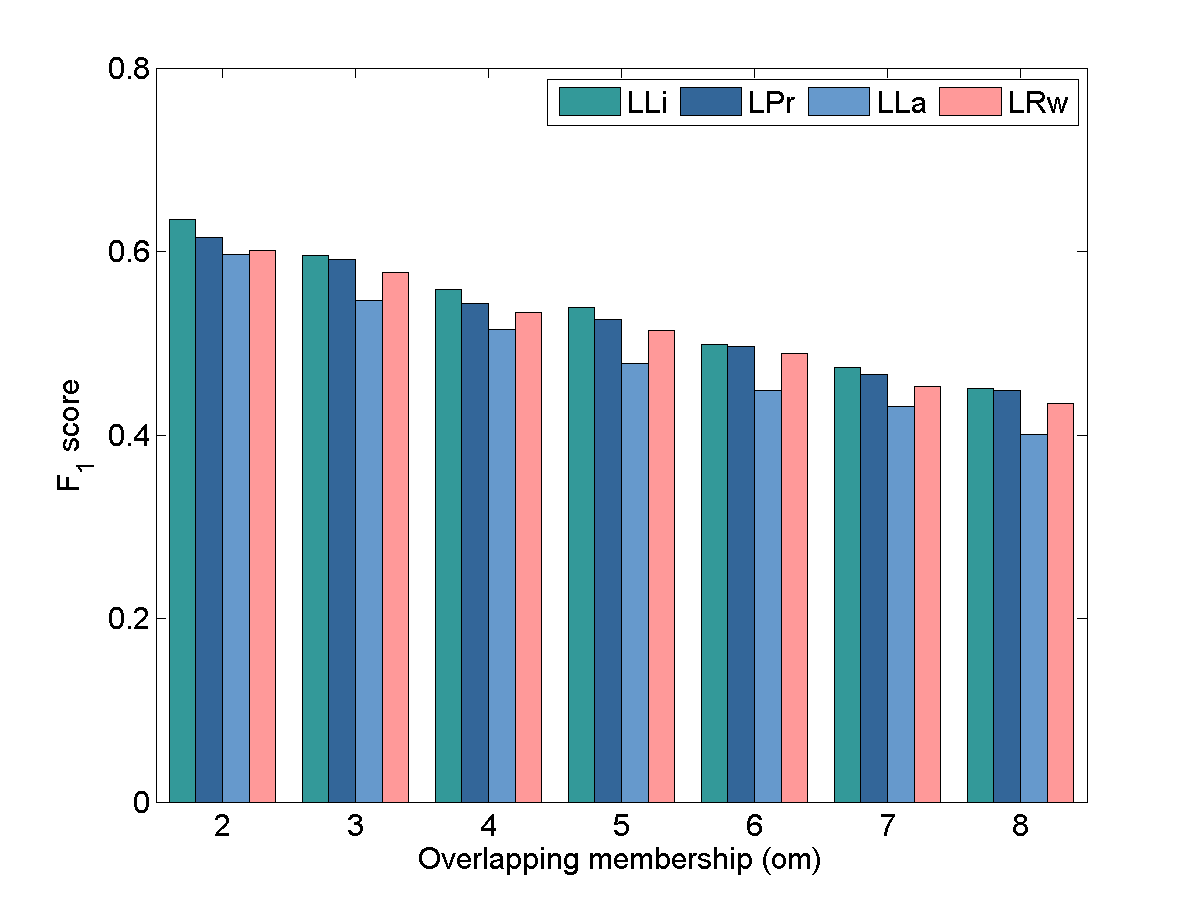}}
		\vspace{-1em}
	\subfigure[\bf{LFR\_b\_0.5}]{
		\includegraphics[width=2.5in]{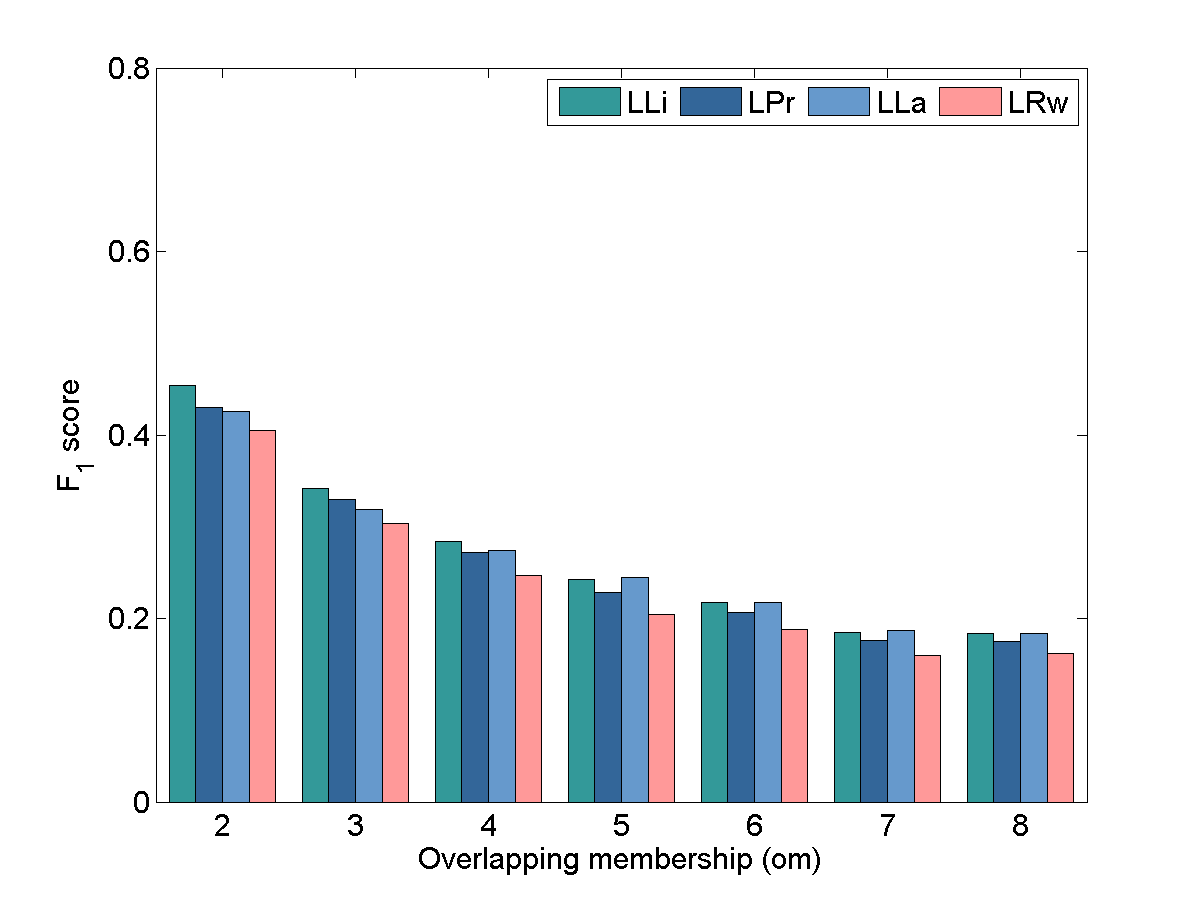}}
\vspace{1em}
	\caption{Accuracy evaluation of LOSP on LFR networks (Defined on $\mathbf{N_{rw}^T}$ Krylov subspace, community size truncated by truth size).
		The accuracy decays when the overlapping membership, $om$, increases from 2 to 8. LOSPs based on the four diffusions show similar accuracy over all datasets. LLi, LOSP with light lazy, demonstrates slightly higher accuracy.}
	\label{Fig:LFRCompare}
\end{figure}

\begin{figure}[!t]
	\centering
	\includegraphics[width=2.5in]{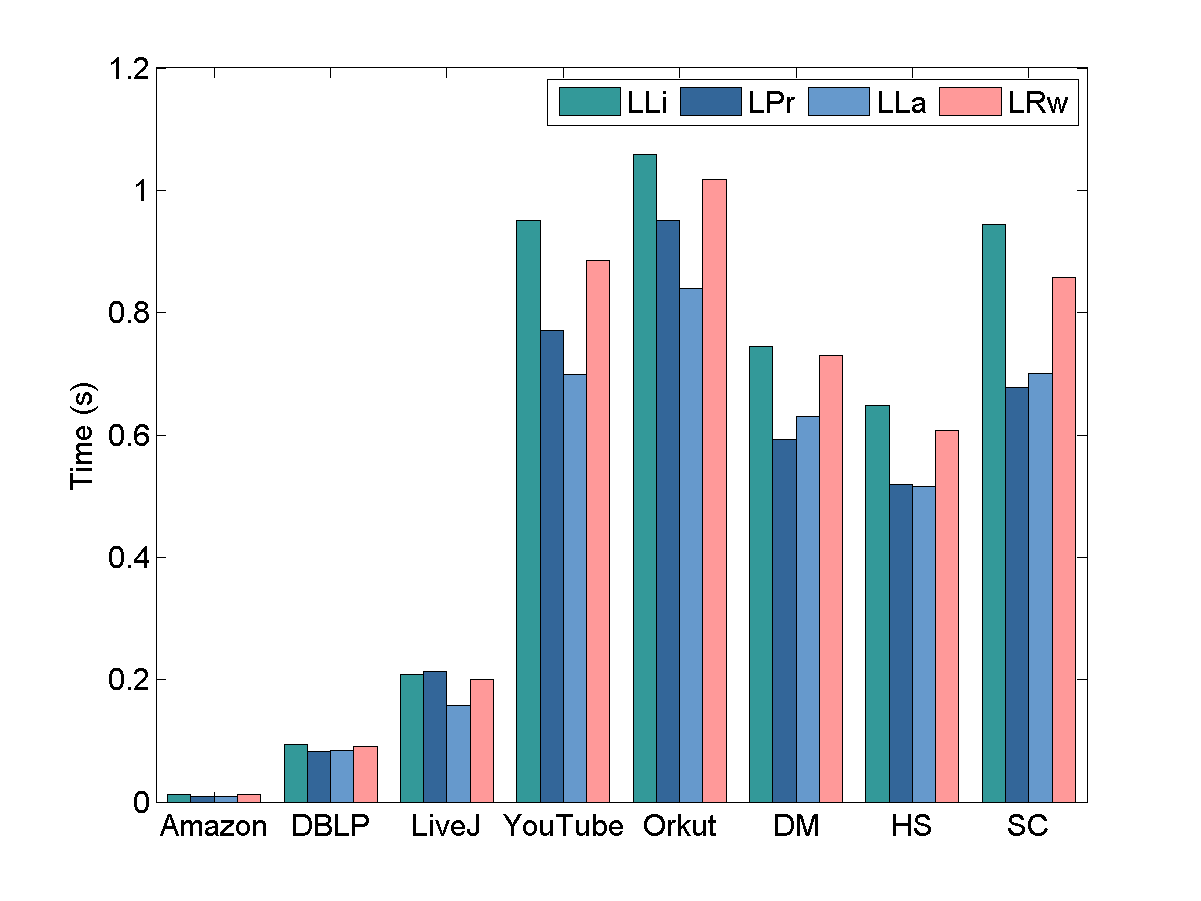}
	\vspace{-1em}
	\caption{Running time evaluation of LOSP on real-world networks (Defined on $\mathbf{N_{rw}^T}$ Krylov subspace, community size truncated by truth size).
		LOSPs based on the four diffusions show similar running time over all datasets, which are within 1.1 seconds.
	}
    \vspace{1.5em}
	\label{Fig:realCompareTime}
\end{figure}

\begin{figure}[!t]
	\vspace{-1em}
	\centering
	\subfigure[\bf{LFR\_s\_0.1}]{
		\includegraphics[width=2.5in]{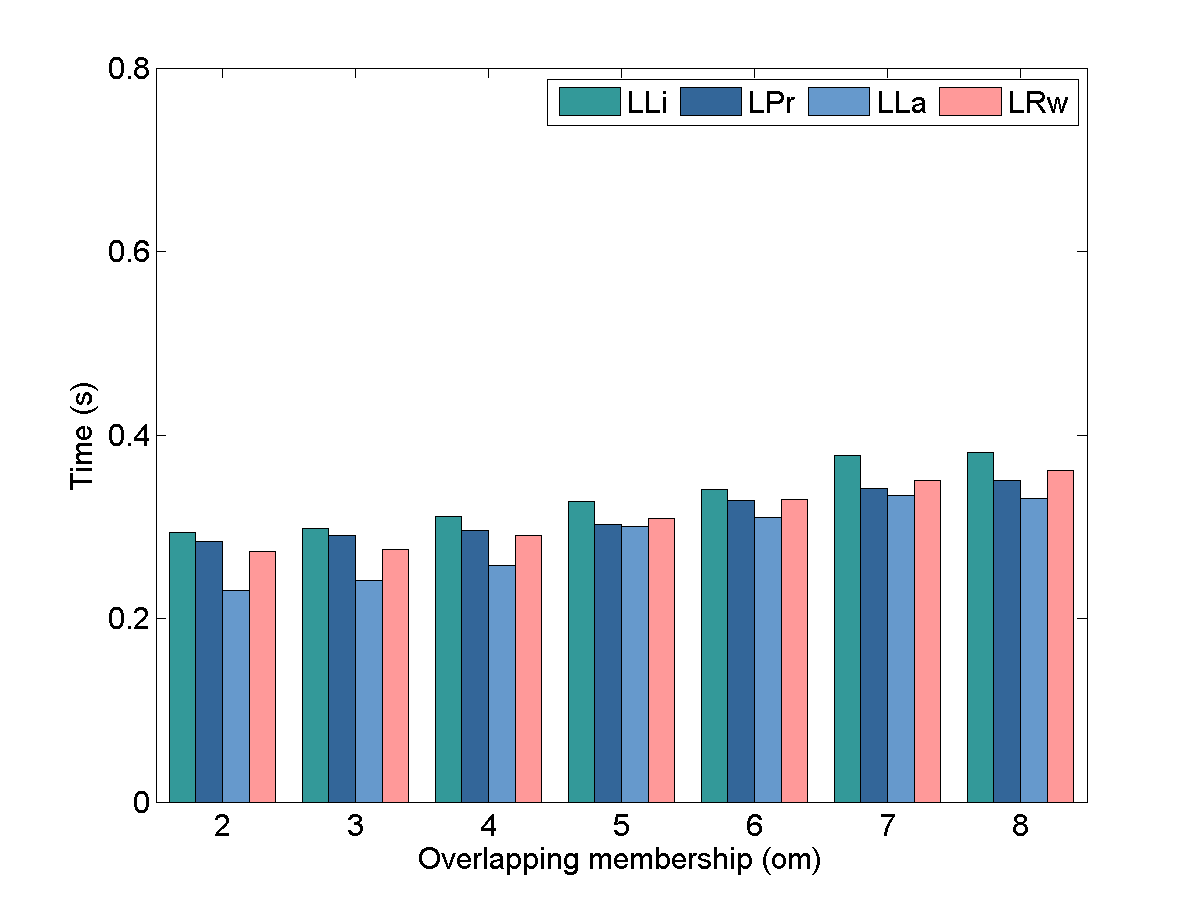}}
		\vspace{-1em}
	\subfigure[\bf{LFR\_s\_0.5}]{
		\includegraphics[width=2.5in]{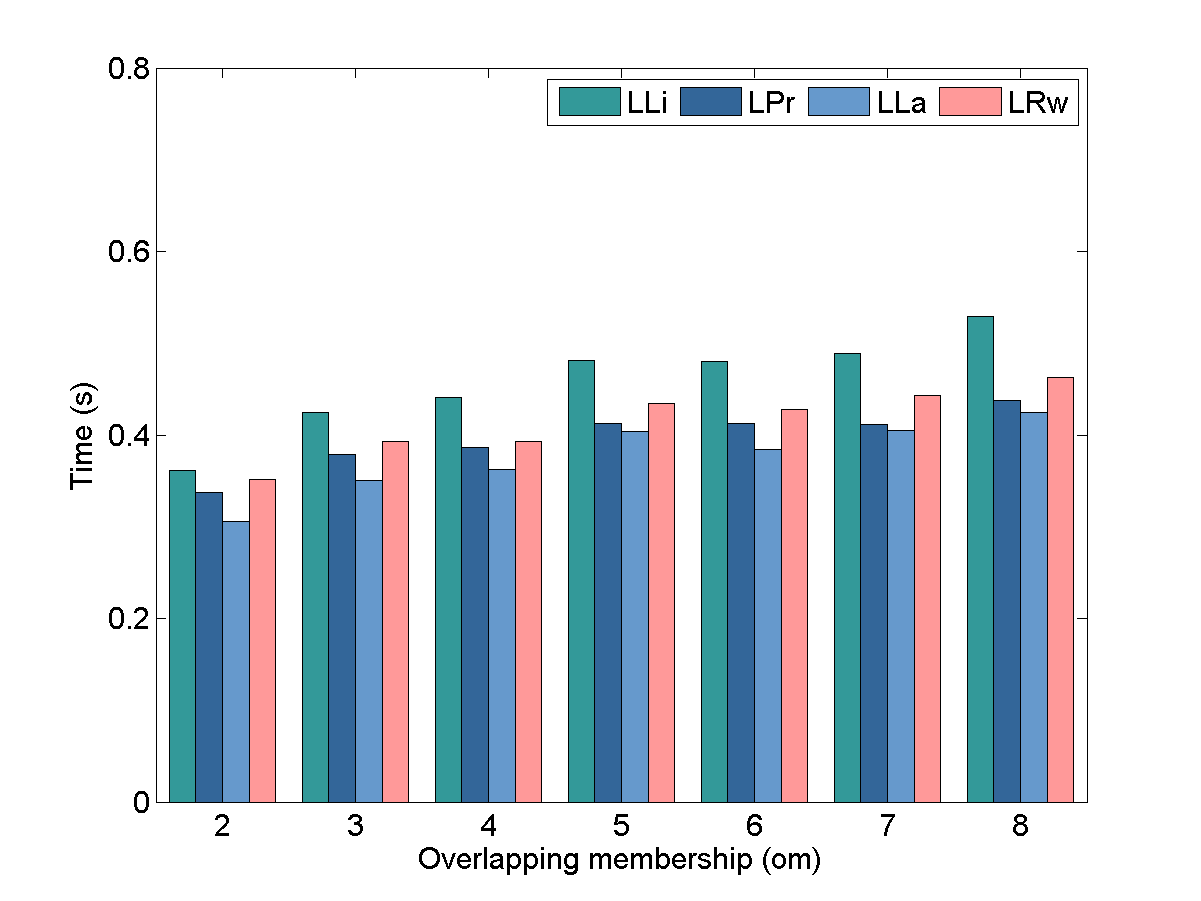}}
	\vspace{-1em}
	\subfigure[\bf{LFR\_b\_0.1}]{
		\includegraphics[width=2.5in]{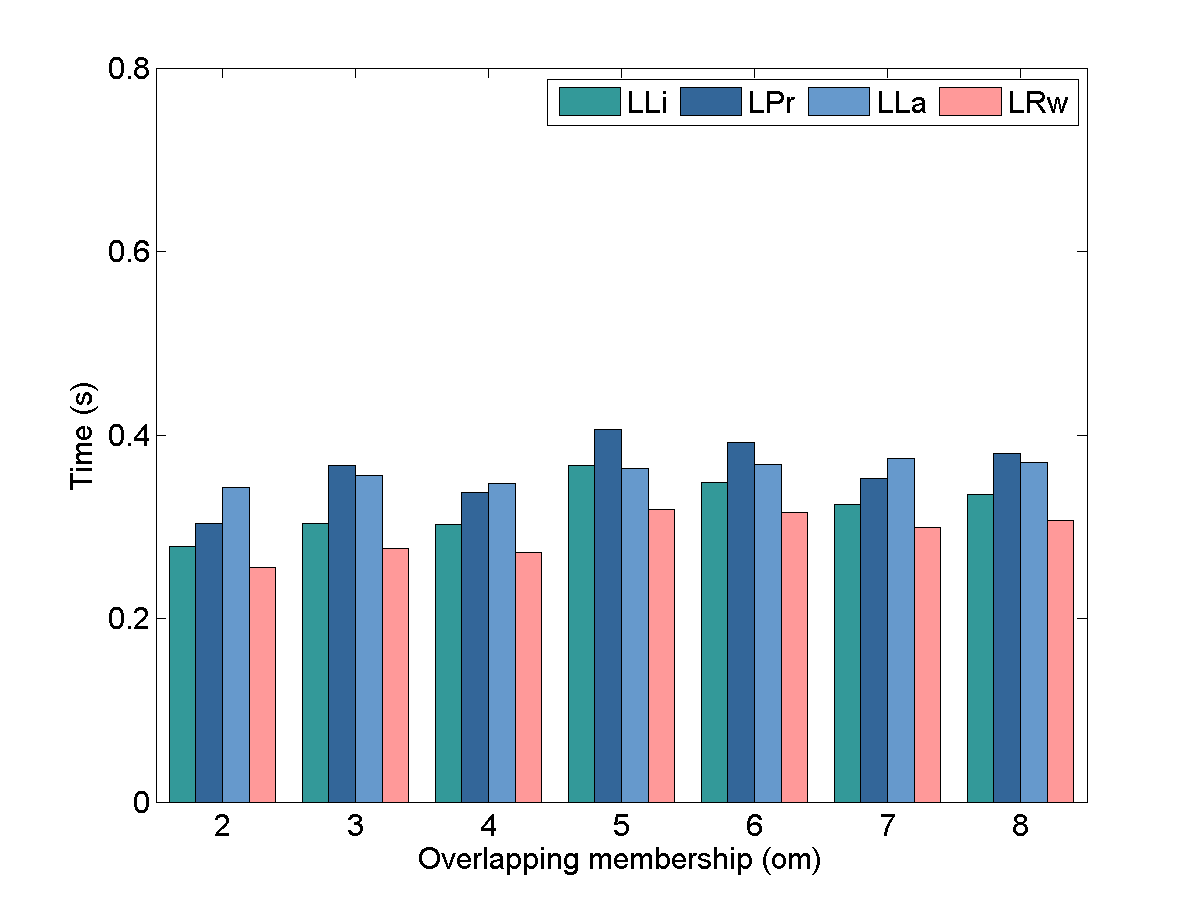}}
		\vspace{-1em}
	\subfigure[\bf{LFR\_b\_0.5}]{
		\includegraphics[width=2.5in]{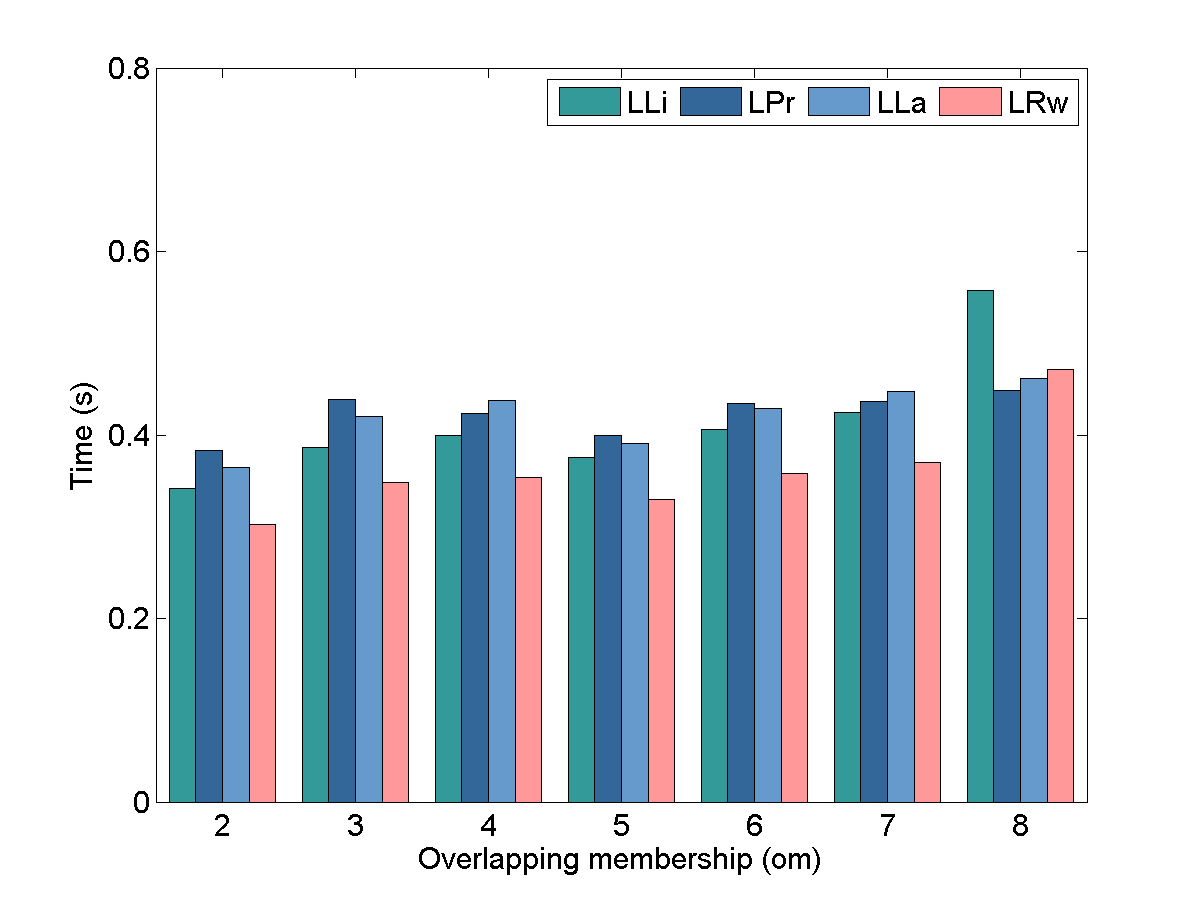}}
    \vspace{1em}
	\caption{Running time evaluation of LOSP on LFR networks (Defined on $\mathbf{N_{rw}^T}$ Krylov subspace, community size truncated by truth size).
		LOSPs based on the four diffusions show similar running time over all datasets, which are around 0.4 seconds.}
    \vspace{1em}
	\label{Fig:LFRCompareTime}
\end{figure}

\subsection{Evaluation on Local Spectral Methods}
\label{sec:LOSPEvaluation}
To remove the impact of different methods in finding a local minimum for the community boundary, we use the ground truth size as a budget for the proposed four LOSP variants: LRw (standard), LLi (light lazy), LLa (lazy) and LPr (pagerank).
We first compare the four LOSP variants defined on the standard $\mathbf{N_{rw}^T}$ Krylov subspace, then compare the general performance on subspaces defined on either $\mathbf{N_{rw}^T}$ or $\mathbf{N_{rw}}$.

\textbf{Evaluation on variants of $\mathbf{N_{rw}^T}$ Krylov subspace.}
Fig. \ref{Fig:realCompare} illustrates the average detection accuracy on the eight real-world datasets. LLi, LPr, LLa and LRw achieve almost the same performance on almost all datasets. One exception is on YouTube that LLi, LLa and LPr considerably outperform LRw. LLi achieves slightly better performance on four out of eight real-world networks.

Fig. \ref{Fig:LFRCompare} illustrates the average detection accuracy on the four sets with a total of 28 LFR networks.
For $on = 500$, Fig. \ref{Fig:LFRCompare} (a) and Fig. \ref{Fig:LFRCompare} (c) show that LLi, LPr and LRw achieve almost the same performance and outperform LLa on average.
For $on = 2500$, Fig. \ref{Fig:LFRCompare} (b) and Fig. \ref{Fig:LFRCompare} (d) show that LLi, LPr and LLa achieve almost the same performance and outperform LRw on average.
On both cases, LLi, LOSP with light lazy, demonstrates slightly higher accuracy on LFR datasets.

For the running time on the sampled subgraphs of all datasets, LLi, LPr, LLa and LRw take almost the same time, and run within 1.1 seconds, as shown in Fig. \ref{Fig:realCompareTime} and Fig. \ref{Fig:LFRCompareTime}.

\textbf{Comparison on $\mathbf{N_{rw}^T}$ and $\mathbf{N_{rw}}$ Krylov subspace.}
We compare the average $F_1$ score on each group of networks: SNAP, Biology and the four groups of LFR. Table \ref{table:SubspaceComparison} shows the comparison of detection accuracy where the community is truncated on truth size for subspace definitions on $\mathbf{N_{rw}^T}$ and $\mathbf{N_{rw}}$ respectively (The best three values on each row appear in bold).
In general, $\mathbf{N_{rw}^T}$ outperforms $\mathbf{N_{rw}}$, especially on real-world datasets.
When comparing with the LOSP variant defined on $\mathbf{N_{rw}}$, the four variants defined on $\mathbf{N_{rw}^T}$ have about 15\% or 5\% higher $F_1$ scores on SNAP and Biology respectively.
As for the synthetic LFR datasets, in general, $\mathbf{N_{rw}^T}$ performs slightly better than $\mathbf{N_{rw}}$ on the small ground truth communities, while $\mathbf{N_{rw}}$ performs slightly better than $\mathbf{N_{rw}^T}$ on the big ground truth communities. Among all the variants of LOSP, LLi on $\mathbf{N_{rw}^T}$ is always on top three for all datasets. In summary, LLi on $\mathbf{N_{rw}^T}$ performs better than other variants of LOSP.
\begin{table}[htbp]
	\renewcommand{\arraystretch}{1.3}
	\caption{Comparison of average $F_1$ score for subspaces on $\mathbf{N_{rw}^T}$ and $\mathbf{N_{rw}}$ (Truncated on truth size). Among all the variants of LOSP, LLi on $\mathbf{N_{rw}^T}$ is always on top three for all datasets.}
	\label{table:SubspaceComparison}
	\centering
	\begin{tabular}{ l | c c c r | c c c r }
		\hline
		\multirow{2}{*}{\bf{Datasets}}  & \multicolumn{4}{c}{\bf{Subspace on} $\mathbf{N_{rw}^T}$} \vline& \multicolumn{4}{c}{\bf{Subspace on} $\mathbf{N_{rw}}$}    \\
		&\bf{LLi} &\bf{LPr} &\bf{LLa} & \bf{LRw}  &\bf{LLi} &\bf{LPr} &\bf{LLa} &\bf{LRw}  \\
		\hline
		\bf{SNAP} & \bf{0.757}  & \bf{0.757}  & \bf{0.749}  & 0.729  & 0.578  & 0.601  & 0.594  & 0.556  \\
		\hline
		\bf{Biology} & \bf{0.328}  & \bf{0.335}  & \bf{0.331}  & 0.315  & 0.287  & 0.284  & 0.246  & 0.266  \\
		\hline
		\bf{LFR\_s\_0.1} & \bf{0.675}  & \bf{0.677}  & 0.570  & 0.660  & \bf{0.673}  & 0.660  & 0.600  & 0.628  \\
		\hline
		\bf{LFR\_s\_0.5}& \bf{0.381}  & \bf{0.371}  & 0.356  & 0.341  & \bf{0.373}  & 0.351  & 0.364  & 0.317  \\
		\hline
		\bf{LFR\_b\_0.1}& \bf{0.536}  & 0.526  & 0.488  & 0.515  & \bf{0.568}  & \bf{0.550}  & 0.516  & 0.523  \\
		\hline
		\bf{LFR\_b\_0.5}&  \bf{0.273}  & 0.260  & 0.265  & 0.238  & \bf{0.272}  & 0.254  & \bf{0.272}  & 0.230  \\
		\hline
	\end{tabular}
\end{table}

\begin{table}[htbp]
	\renewcommand{\arraystretch}{1.3}
	\caption{Comparison of average $F_1$ score for subspaces on $\mathbf{N_{rw}^T}$ and $\mathbf{N_{rw}}$ (Truncated on local minimal conductance). Among all the variants of LOSP, LLi on $\mathbf{N_{rw}^T}$ is always on top three for all datasets.}
	\label{table:SubspaceComparisonCond}
	\centering
	\begin{tabular}{ l | c c c r | c c c r }
		\hline
		\multirow{2}{*}{\bf{Datasets}}  & \multicolumn{4}{c}{\bf{Subspace on} $\mathbf{N_{rw}^T}$} \vline& \multicolumn{4}{c}{\bf{Subspace on} $\mathbf{N_{rw}}$}    \\
		&\bf{LLi} &\bf{LPr} &\bf{LLa} & \bf{LRw}  &\bf{LLi} &\bf{LPr} &\bf{LLa} &\bf{LRw}  \\
		\hline
		\bf{SNAP} & \bf{0.620}  & \bf{0.604}  & \bf{0.599}  & 0.593  & 0.518  & 0.514  & 0.526  & 0.503  \\
		\hline
		\bf{Biology} & \bf{0.182}  & \bf{0.171}  & \bf{0.206}  & 0.150  & 0.146  & 0.160  & 0.166  & 0.135  \\
		\hline
		\bf{LFR\_s\_0.1} & \bf{0.574}  & \bf{0.561}  & 0.443  & 0.554  & \bf{0.576}  & 0.546  & 0.518  & 0.539  \\
		\hline
		\bf{LFR\_s\_0.5}& \bf{0.334}  & 0.291  & 0.265  & 0.232  & \bf{0.332}  & 0.311  & \bf{0.337}  & 0.217  \\
		\hline
		\bf{LFR\_b\_0.1}& \bf{0.387}  & 0.365  & 0.296  & 0.358  & \bf{0.406}  & 0.348  & 0.353  & \bf{0.393}  \\
		\hline
		\bf{LFR\_b\_0.5}& \bf{0.198}  & 0.169  & 0.168  & 0.134  & \bf{0.217}  & 0.186  & \bf{0.214}  & 0.131  \\
		\hline
	\end{tabular}
\end{table}

It is interesting that the results on $\mathbf{N_{rw}}$ is reasonably good. Even if we use the first local conductance to do the truncation, as shown in Table \ref{table:SubspaceComparisonCond} (The best three values on each row appear in bold), the accuracy decays by about 5\% to 15\% as compared with that of the truncation on truth size (in Table \ref{table:SubspaceComparison}). Nevertheless, the results of Table \ref{table:SubspaceComparisonCond} are still much better than that of state-of-the-art local community detection algorithms, heat kernel (HK) and pagerank (PR), whose results are listed in Table \ref{table: Realcomparison} and Table \ref{table: LFRcomparison}.

\begin{figure}[!t]
	\vspace{-1em}
	\centering
	\subfigure[]{
		\includegraphics[width=2.65in]{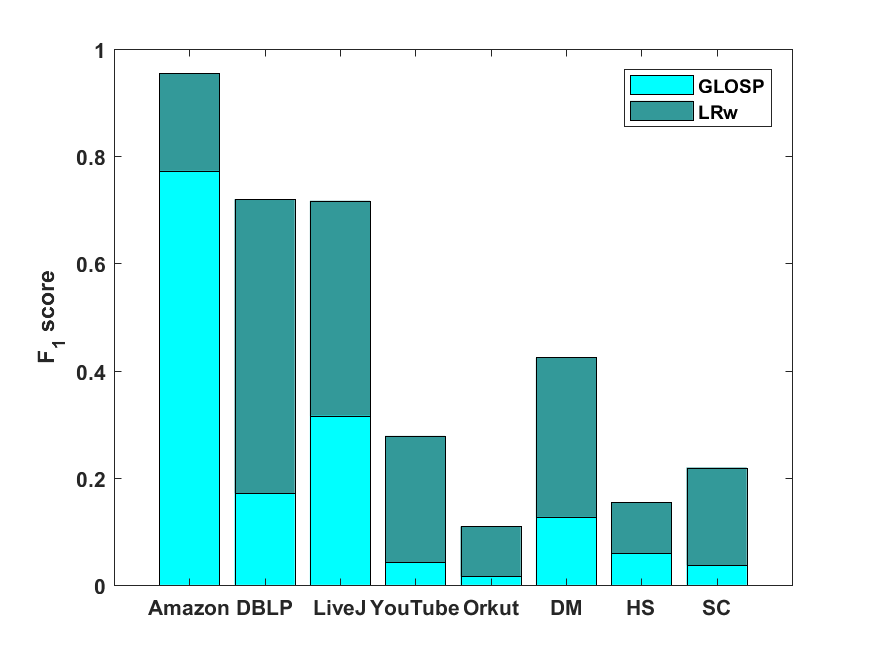}}
		\vspace{-1em}
	\subfigure[]{
		\includegraphics[width=2.65in]{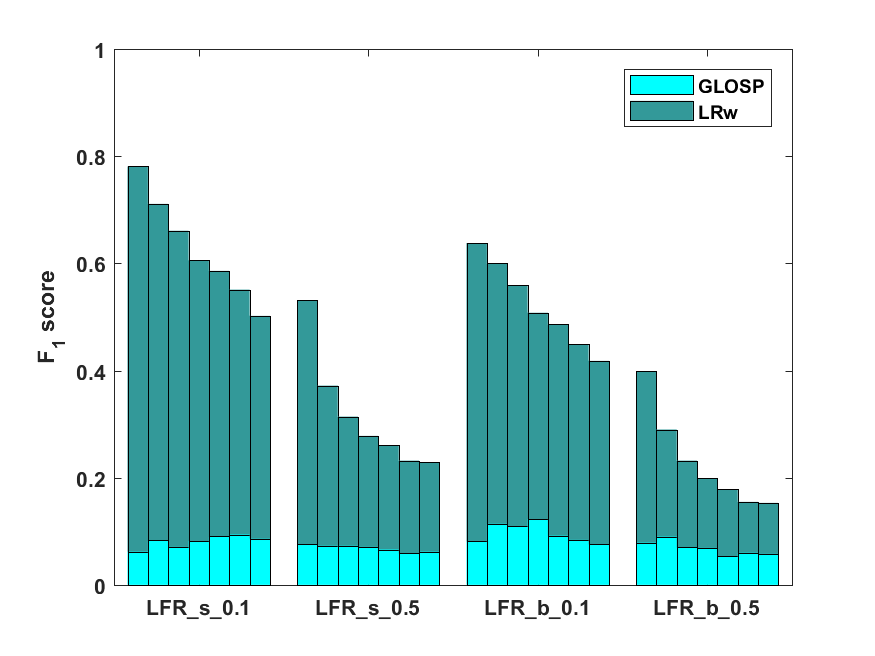}}
	\caption{Accuracy comparison with GLOSP based on actual eigenspace on real-world networks and LFR datasets (Community size truncated by truth size). In (b), the overlapping membership ($om$) starts from 2 to 8 for each group of LFR datasets. Note that the bars for GLOSP and LRw both start from zero to the height.}
	\label{Fig:GlobalCompare}
\end{figure}

\textbf{Comparison on Krylov subspace and eigenspace for $\mathbf{N_{rw}}$.}
In order to evaluate the effectiveness of Krylov subspace approximation, we compare Krylov subspace with the actual eigenspace. To make a fair comparison, we choose LRw defined on $\mathbf{N_{rw}}$ Krylov subspace and adopt our method on the actual eigenspace associated with two leading eigenvectors of standard transition matrix $\mathbf{N_{rw}}$ with larger eigenvalues, denoted as GLOSP (GLObal SPectral).

Fig. \ref{Fig:GlobalCompare} illustrates the average detection accuracy on real-world networks and LFR datasets. For real-world networks, Fig. \ref{Fig:GlobalCompare} (a) shows that LRw considerably outperforms GLOSP. For LFR datasets, Fig. \ref{Fig:GlobalCompare} (b) illustrates that LRw is much better than GLOSP on each dataset. These experiments show that it is not good if we use the actual eigenspace that embodies the global structure, even on the comparatively small subgraph.

\subsection{Final Comparison}
\label{sec:finalCompare}
For the final comparison, we find a local minimum of conductance to automatically determine the community boundary.
Subsection \ref{sec:LOSPEvaluation} shows that LLi, LOSP with light lazy, is on average the best LOSP method defined on $\mathbf{N_{rw}^T}$ Krylov subspace on all datasets. In this subsection, we just compare LLi defined on the $\mathbf{N_{rw}^T}$ Krylov subspace with four state-of-the-art local community detection algorithms, LEMON \cite{LiHK_TKDD18}, PGDc-d \cite{LaarhovenJMLR16}, HK for \texttt{hk-relax}~\cite{kloster2014heat} and PR for \texttt{pprpush}~\cite{pagerank2006}, which also use conductance as the metric to determine the community boundary. LEMON is a local spectral approach based on normalized adjacency matrix iteration, PGDc-d is a projected gradient descent algorithm for optimizing $\sigma$-conductance, and HK is based on heat kernel diffusion while PR is based on the pagerank diffusion.
To make a fair comparison, we use the default parameter settings for baselines as they also test on SNAP datasets, and run the five algorithms on the same three seeds randomly chosen from the ground truth communities.

\subsubsection{Comparison on Real-world Datasets}
For each of the real-world networks, Table \ref{table: Realcomparison} shows the average detection accuracy (The best value appears in bold in each row) and the average running time,
Table \ref{table: Realcomparison2} shows the average community size and the average conductance of the detected communities.

For the SNAP datasets in product, collaboration and social domains, LLi yields considerably higher accuracy on the first three datasets (Amazon, DBLP, and LiveJ), 
and yields slightly lower accuracy on YouTube and Orkut. To have a better understanding on the results, we compare the property of the detection with the property of the ground truth communities, as shown in Table \ref{table:realdata_stats}.
\begin{itemize}
	\item For the first three datasets, the average size of the ground truth is small between 10 to 30, and the conductance is diverse (very low in Amazon and with reasonable value around 0.4 in DBLP and LiveJ). In general, the size and conductance of our detected communities are closest to the ground truth. This may explain why our detection accuracy is considerably higher than the baselines.
	\item For the last two datasets YouTube and Orkut, 
the average conductance of ground truth is very high with 0.84 and 0.73, respectively, and the size is diverse (very small of 21 on average in YouTube and large of 216 on average in Orkut).
	We found much larger communities with close conductance on YouTube (0.736) and Orkut (0.791). PGDc-d found much smaller communities with close conductance on YouTube, and HK found reasonable-size communities with lower conductance on Orkut.
	This may explain why our detection accuracy on YouTube and Orkut is slightly lower than PGDc-d and HK, respectively.	
\end{itemize}

For the three Biology datasets, LLi outperforms HK and PR, and yields slightly lower accuracy compared with LEMON and PGDc-d. Also, the conductance of the communities detected by LEMON (around 0.88) and PGDc-d (around 0.98) are very close to that of the ground truth (around 0.88, as shown in Table \ref{table:realdata_stats}), compared with LLi (around 0.74), HK (around 0.37) and PR (around 0.20). As for the community size, we detected larger communities as compared to the ground truth, LEMON and PGDc-d detect much smaller communities, while HK and PR detect even larger communities.

For the running time, as shown in Table \ref{table: Realcomparison}, LLi is fast in 4 seconds on SNAP and in 16 seconds on Biology datasets.
The four baselines, LEMON, PGDc-d, HK and PR, are faster in less than 1.2 seconds. This may due to that LLi is implemented in Matlab while the most baselines are implemented in C++.

\begin{table*}[htbp]
\renewcommand{\arraystretch}{1.3}
\caption{Comparison on detection accuracy and running time with baselines on real-world networks.}
\label{table: Realcomparison}
\centering
\scalebox{0.8}{
\begin{tabular}{ l l | c c c c  r | c c c c  r }
\hline
& \multirow{2}{*}{\bf{Datasets}} & \multicolumn{5}{c}{\bf{$F_1$ score}} \vline& \multicolumn{5}{c}{\bf{Time (s)}}   \\
& &\bf{LLi} &\bf{LEMON} &\bf{PGDc-d} &\bf{HK} & \bf{PR} &\bf{LLi} &\bf{LEMON} &\bf{PGDc-d} &\bf{HK} & \bf{PR}   \\
 \hline
\bf{SNAP} &\bf{Amazon} & \bf{0.800}  & 0.723  & 0.576  & 0.751  & 0.531  & 0.011  & 0.039  & 0.058  & 0.008  & 0.015  \\
&\bf{DBLP}  & \bf{0.779}  & 0.590  & 0.523  & 0.413  & 0.388  & 0.099  & 0.271  & 0.063  & 0.025  & 0.075  \\
&\bf{LiveJ} & \bf{0.797}  & 0.519  & 0.388  & 0.573  & 0.525  & 0.756  & 0.553  & 1.759  & 0.029  & 0.264  \\
&\bf{YouTube} & 0.445  & 0.351  & \bf{0.447}  & 0.091  & 0.143  & 9.268  & 1.458  & 0.329  & 0.038  & 0.955  \\
&\bf{Orkut} & 0.279  & 0.096  & 0.117  & \bf{0.357}  & 0.303  & 12.235  & 1.328  & 3.212  & 0.027  & 1.392  \\
\hline
&\bf{Average} & \textbf{0.620 } & 0.456  & 0.410  & 0.437 & 0.378  & 4.474 & 0.730 & 1.084 & 0.025 & 0.540 \\

\hline
\bf{Biology} &\bf{DM} & 0.221  & \bf{0.322}  & 0.274  & 0.219  & 0.035  & 33.535  & 1.151  & 0.007  & 0.015  & 1.394  \\
&\bf{HS} & 0.206  & 0.200  & \bf{0.213}  & 0.025  & 0.020  & 3.823  & 0.462  & 0.003  & 0.071  & 1.048  \\
&\bf{SC} & 0.120  & \bf{0.239}  & 0.238  & 0.036  & 0.034  & 11.119  & 0.950  & 0.002  & 0.034  & 0.979  \\
\hline
&\bf{Average} & 0.182 & \textbf{0.254 } & 0.242 & 0.093 & 0.030 & 16.159 & 0.854 & 0.004 & 0.040 & 1.140 \\
\hline
\end{tabular}}
\end{table*}

\begin{table*}[htbp]
\renewcommand{\arraystretch}{1.3}
\caption{Comparison on community size and conductance with baselines on real-world networks. }
\label{table: Realcomparison2}
\centering
\scalebox{0.8}{
\begin{tabular}{ l l | c c c c r | c c c c r }
\hline
 &    \multirow{2}{*}{\bf{Datasets}}      & \multicolumn{5}{c}{\bf{Size}} \vline& \multicolumn{5}{c}{\bf{Conductance}}   \\
 &  &\bf{LLi} &\bf{LEMON} &\bf{PGDc-d} &\bf{HK} & \bf{PR} &\bf{LLi} &\bf{LEMON} &\bf{PGDc-d} &\bf{HK} & \bf{PR}   \\
 \hline
\bf{SNAP} &\bf{Amazon} & 8     & 13    & 3     & 48    & 4485  & 0.300  & 0.850  & 0.605  & 0.042  & 0.030  \\
&\bf{DBLP}  & 10    & 20    & 4     & 87    & 9077  & 0.396  & 0.559  & 0.751  & 0.110  & 0.114  \\
&\bf{LiveJ} & 35    & 26    & 3     & 119   & 512   & 0.305  & 0.637  & 0.817  & 0.083  & 0.086  \\
&\bf{YouTube} & 588   & 18    & 12    & 122   & 13840  & 0.736  & 0.748  & 0.723  & 0.175  & 0.302  \\
&\bf{Orkut} & 1733  & 9     & 3     & 341   & 1648  & 0.791  & 0.934  & 0.976  & 0.513  & 0.546  \\
\hline
&\bf{Average} & 475 & 17 & 5 & 143 & 5912 & 0.506  & 0.746 & 0.774 & 0.185 & 0.216 \\
\hline
\bf{Biology} &\bf{DM} & 992   & 19    & 3     & 606   & 15120  & 0.772  & 0.909  & 0.980  & 0.571  & 0.181  \\
&\bf{HS} & 465   & 13    & 5     & 4360  & 10144  & 0.758  & 0.836  & 0.966  & 0.229  & 0.080  \\
&\bf{SC} & 1623  & 14    & 3     & 2604  & 2673  & 0.696  & 0.903  & 0.990  & 0.322  & 0.343  \\
\hline
&\bf{Average} & 1027 & 15 & 4 & 2523 & 9312 & 0.742 & 0.883 & 0.979 & 0.374 & 0.201 \\
\hline
\end{tabular} }
\end{table*}


\begin{table*}[htbp]
\renewcommand{\arraystretch}{1.3}
\caption{Comparison on detection accuracy and running time with baselines on LFR networks.}
\label{table: LFRcomparison}
\centering
\scalebox{0.8}{
\begin{tabular}{ l l | c c c c r | c c c c r }
\hline
&   \multirow{2}{*}{\bf{Datasets}}     & \multicolumn{5}{c}{\bf{$F_1$ score}} \vline& \multicolumn{5}{c}{\bf{Time (s)}}   \\
&  &\bf{LLi} &\bf{LEMON} &\bf{PGDc-d} &\bf{HK} & \bf{PR} &\bf{LLi} &\bf{LEMON} &\bf{PGDc-d} &\bf{HK} & \bf{PR}   \\
\hline
\bf{LFR} &\bf{s\_0.1\_2}   &    \bf{0.673}  & 0.458  & 0.261  & 0.051  & 0.025  & 0.918  & 1.282  & 0.207  & 0.025  & 0.692  \\
&\bf{s\_0.1\_3} &    \bf{0.621}  & 0.452  & 0.261  & 0.028  & 0.025  & 0.969  & 1.031  & 0.201  & 0.025  & 0.693  \\
&\bf{s\_0.1\_4} &    \bf{0.583}  & 0.445  & 0.263  & 0.027  & 0.024  & 1.022  & 1.470  & 0.201  & 0.025  & 0.690  \\
&\bf{s\_0.1\_5} &    \bf{0.555}  & 0.436  & 0.258  & 0.044  & 0.024  & 1.123  & 1.814  & 0.202  & 0.026  & 0.675  \\
&\bf{s\_0.1\_6} &    \bf{0.548}  & 0.435  & 0.255  & 0.026  & 0.024  & 1.269  & 1.666  & 0.207  & 0.025  & 0.682  \\
&\bf{s\_0.1\_7} &    \bf{0.536}  & 0.441  & 0.260  & 0.041  & 0.026  & 1.342  & 1.703  & 0.197  & 0.025  & 0.677  \\
&\bf{s\_0.1\_8} &    \bf{0.503}  & 0.429  & 0.252  & 0.029  & 0.024  & 1.371  & 1.502  & 0.205  & 0.025  & 0.678  \\
\hline
&\bf{Average} &    \textbf{0.574 } & 0.442 & 0.259 & 0.035 & 0.025 & 1.145 & 1.495 & 0.203 & 0.025 & 0.684 \\
\hline
&\bf{s\_0.5\_2} &    \bf{0.496}  & 0.389  & 0.266  & 0.019  & 0.020  & 1.232  & 1.222  & 0.209  & 0.028  & 0.680  \\
&\bf{s\_0.5\_3} &    \bf{0.388}  & 0.310  & 0.253  & 0.020  & 0.022  & 1.626  & 1.160  & 0.212  & 0.026  & 0.686  \\
&\bf{s\_0.5\_4} &    \bf{0.336}  & 0.295  & 0.257  & 0.019  & 0.020  & 1.708  & 1.032  & 0.232  & 0.028  & 0.683  \\
&\bf{s\_0.5\_5} &    \bf{0.313}  & 0.281  & 0.261  & 0.018  & 0.019  & 1.949  & 1.052  & 0.217  & 0.028  & 0.687  \\
&\bf{s\_0.5\_6} &    \bf{0.284}  & 0.268  & 0.257  & 0.017  & 0.018  & 1.899  & 0.896  & 0.214  & 0.027  & 0.690  \\
&\bf{s\_0.5\_7} &    \bf{0.260}  & 0.249  & 0.257  & 0.017  & 0.018  & 1.914  & 0.928  & 0.216  & 0.027  & 0.694  \\
&\bf{s\_0.5\_8} &    0.260  & 0.247  & \bf{0.261}  & 0.016  & 0.017  & 2.037  & 0.939  & 0.216  & 0.025  & 0.696  \\
\hline
&\bf{Average} &     \textbf{0.334 } & 0.291 & 0.259 & 0.018 & 0.019  & 1.766  & 1.033  & 0.217  & 0.027  & 0.688 \\
\hline
&\bf{b\_0.1\_2} &    \bf{0.461}  & 0.279  & 0.140  & 0.094  & 0.040  & 1.089  & 0.924  & 0.189  & 0.026  & 0.737  \\
&\bf{b\_0.1\_3} &    \bf{0.417}  & 0.260  & 0.135  & 0.064  & 0.043  & 1.253  & 0.779  & 0.189  & 0.025  & 0.720  \\
&\bf{b\_0.1\_4} &    \bf{0.392}  & 0.268  & 0.141  & 0.092  & 0.039  & 1.424  & 0.931  & 0.195  & 0.025  & 0.717  \\
&\bf{b\_0.1\_5} &    \bf{0.361}  & 0.255  & 0.134  & 0.073  & 0.042  & 1.595  & 0.963  & 0.194  & 0.025  & 0.713  \\
&\bf{b\_0.1\_6} &    \bf{0.374}  & 0.272  & 0.144  & 0.060  & 0.038  & 1.640  & 0.923  & 0.194  & 0.025  & 0.718  \\
&\bf{b\_0.1\_7} &    \bf{0.353}  & 0.267  & 0.143  & 0.047  & 0.039  & 1.608  & 1.019  & 0.207  & 0.025  & 0.714  \\
&\bf{b\_0.1\_8} &    \bf{0.349}  & 0.269  & 0.143  & 0.051  & 0.038  & 1.597  & 0.902  & 0.209  & 0.025  & 0.721  \\
\hline
&\bf{Average} &    \textbf{0.387 } & 0.267 & 0.140 & 0.069 & 0.040 & 1.458  & 0.920  & 0.197  & 0.025  & 0.720 \\
\hline
&\bf{b\_0.5\_2} &    \bf{0.309}  & 0.207  & 0.135  & 0.040  & 0.042  & 1.579  & 0.793  & 0.210  & 0.025  & 0.718  \\
&\bf{b\_0.5\_3} &    \bf{0.236}  & 0.197  & 0.135  & 0.039  & 0.041  & 1.858  & 0.896  & 0.205  & 0.025  & 0.706  \\
&\bf{b\_0.5\_4} &    \bf{0.198}  & 0.176  & 0.138  & 0.035  & 0.037  & 2.026  & 0.883  & 0.200  & 0.025  & 0.701  \\
&\bf{b\_0.5\_5} &    \bf{0.177}  & 0.154  & 0.135  & 0.033  & 0.036  & 1.818  & 0.949  & 0.205  & 0.025  & 0.689  \\
&\bf{b\_0.5\_6} &    \bf{0.170}  & 0.151  & 0.141  & 0.031  & 0.033  & 2.079  & 1.029  & 0.216  & 0.025  & 0.705  \\
&\bf{b\_0.5\_7} &    \bf{0.150}  & 0.144  & 0.134  & 0.031  & 0.033  & 2.123  & 1.142  & 0.200  & 0.025  & 0.704  \\
&\bf{b\_0.5\_8} &    \bf{0.149}  & 0.149  & 0.137  & 0.030  & 0.032  & 2.116  & 0.968  & 0.203  & 0.025  & 0.703  \\
\hline
&\bf{Average} &    \textbf{0.198 } & 0.168 & 0.136 & 0.034 & 0.036 & 1.943  & 0.951  & 0.206  & 0.025  & 0.704 \\
\hline
\end{tabular} }
\end{table*}

\begin{table*}[htbp]
\renewcommand{\arraystretch}{1.3}
\caption{Comparison on community size and conductance with baselines on LFR networks.}
\label{table: LFRcomparison2}
\centering
\scalebox{0.8}{
\begin{tabular}{ l l | c c c c r | c c c c r }
\hline
 &    \multirow{2}{*}{\bf{Datasets}}      & \multicolumn{5}{c}{\bf{Size}} \vline& \multicolumn{5}{c}{\bf{Conductance}}   \\
&  &\bf{LLi} &\bf{LEMON} &\bf{PGDc-d} &\bf{HK} & \bf{PR} &\bf{LLi} &\bf{LEMON} &\bf{PGDc-d} &\bf{HK} & \bf{PR}   \\
 \hline
\bf{LFR}&    \textbf{s\_0.1\_2} & 13    & 20    & 3     & 2279  & 2327  & 0.557  & 0.537  & 0.938  & 0.276  & 0.269  \\
Size:&    \textbf{s\_0.1\_3} & 12    & 17    & 3     & 2356  & 2317  &  0.598  & 0.563  & 0.945  & 0.289  & 0.281 \\
$[10,50]$&    \textbf{s\_0.1\_4} & 11    & 20    & 3     & 2354  & 2318  &  0.619  & 0.534  & 0.954  & 0.294  & 0.286 \\
&    \textbf{s\_0.1\_5} & 11    & 19    & 3     & 2363  & 2328  & 0.636  & 0.556  & 0.957  & 0.300  & 0.290  \\
Cond.:&    \textbf{s\_0.1\_6} & 34    & 20    & 3     & 2387  & 2302  & 0.627  & 0.548  & 0.954  & 0.295  & 0.294  \\
0.522&    \textbf{s\_0.1\_7} & 22    & 20    & 3     & 2349  & 2272  &  0.644  & 0.555  & 0.959  & 0.300  & 0.297 \\
&    \textbf{s\_0.1\_8} & 39    & 19    & 3     & 2371  & 2278  & 0.666  & 0.568  & 0.968  & 0.305  & 0.302  \\
     \hline
&    \textbf{Average} & 20 & 19 & 3 & 2351 & 2306 & 0.621 & 0.552 & 0.954 & 0.294 & 0.288 \\
     \hline
&    \textbf{s\_0.5\_2} & 49    & 14    & 3     & 2420  & 2313  &  0.741  & 0.643  & 0.956  & 0.333  & 0.324 \\
Size:&    \textbf{s\_0.5\_3} & 76    & 15    & 3     & 2456  & 2281  & 0.789  & 0.675  & 0.966  & 0.348  & 0.346  \\
$[10,50]$&    \textbf{s\_0.5\_4} & 121   & 15    & 3     & 2459  & 2258  & 0.801  & 0.695  & 0.972  & 0.359  & 0.357  \\
&    \textbf{s\_0.5\_5} & 129   & 16    & 3     & 2463  & 2246  & 0.814  & 0.699  & 0.978  & 0.362  & 0.363  \\
Cond.:&    \textbf{s\_0.5\_6} & 210   & 16    & 3     & 2384  & 2242  & 0.813  & 0.714  & 0.981  & 0.355  & 0.368  \\
0.746&    \textbf{s\_0.5\_7} & 213   & 17    & 3     & 2374  & 2231  & 0.817  & 0.715  & 0.986  & 0.358  & 0.371  \\
&    \textbf{s\_0.5\_8} & 220   & 17    & 3     & 2304  & 2230  & 0.812  & 0.721  & 0.985  & 0.368  & 0.374  \\
     \hline
&    \textbf{Average} & 145 & 16 & 3 & 2409 & 2257 & 0.798 & 0.695 & 0.975 & 0.355 & 0.358  \\
     \hline
&    \textbf{b\_0.1\_2} & 16    & 16    & 3     & 2258  & 2344  & 0.666  & 0.660  & 0.977  & 0.303  & 0.286  \\
Size:&    \textbf{b\_0.1\_3} & 15    & 14    & 3     & 2368  & 2340  & 0.688  & 0.686  & 0.975  & 0.307  & 0.295  \\
$[20,100]$&    \textbf{b\_0.1\_4} & 22    & 16    & 3     & 2267  & 2343  & 0.702  & 0.656  & 0.981  & 0.308  & 0.295  \\
&    \textbf{b\_0.1\_5} & 25    & 18    & 3     & 2354  & 2334  & 0.716  & 0.659  & 0.976  & 0.301  & 0.300  \\
Cond.:&    \textbf{b\_0.1\_6} & 50    & 17    & 3     & 2374  & 2323  & 0.698  & 0.647  & 0.974  & 0.296  & 0.301  \\
0.497&    \textbf{b\_0.1\_7} & 29    & 17    & 3     & 2417  & 2326  &  0.710  & 0.638  & 0.980  & 0.296  & 0.301 \\
&    \textbf{b\_0.1\_8} & 40    & 16    & 3     & 2374  & 2330  & 0.714  & 0.664  & 0.980  & 0.312  & 0.304  \\
     \hline
&    \textbf{Average} & 28  & 16  & 3  & 2344  & 2334  & 0.699 & 0.659 & 0.978 & 0.303 & 0.297  \\
     \hline
&    \textbf{b\_0.5\_2} & 33    & 14    & 3     & 2465  & 2299  &  0.792  & 0.711  & 0.982  & 0.343  & 0.336 \\
Size:&    \textbf{b\_0.5\_3} & 106   & 16    & 3     & 2430  & 2280  & 0.808  & 0.695  & 0.986  & 0.340  & 0.355  \\
$[20,100]$&    \textbf{b\_0.5\_4} & 237   & 16    & 3     & 2440  & 2252  & 0.813  & 0.711  & 0.990  & 0.353  & 0.362  \\
&    \textbf{b\_0.5\_5} & 184   & 18    & 3     & 2473  & 2242  & 0.815  & 0.692  & 0.989  & 0.373  & 0.365  \\
Cond.:&    \textbf{b\_0.5\_6} & 196   & 15    & 3     & 2425  & 2240  &  0.817  & 0.733  & 0.991  & 0.359  & 0.368 \\
0.733&    \textbf{b\_0.5\_7} & 247   & 19    & 3     & 2415  & 2224  & 0.810  & 0.696  & 0.994  & 0.363  & 0.369  \\
&    \textbf{b\_0.5\_8} & 268   & 17    & 3     & 2406  & 2221  & 0.804  & 0.720  & 0.994  & 0.363  & 0.370  \\
     \hline
&    \textbf{Average} & 182  & 16  & 3  & 2436  & 2251  & 0.808 & 0.708 & 0.989 & 0.356 & 0.361  \\
     \hline
\end{tabular} }
\end{table*}

\begin{figure}[!t]
	\centering
	\includegraphics[width=2.65in]{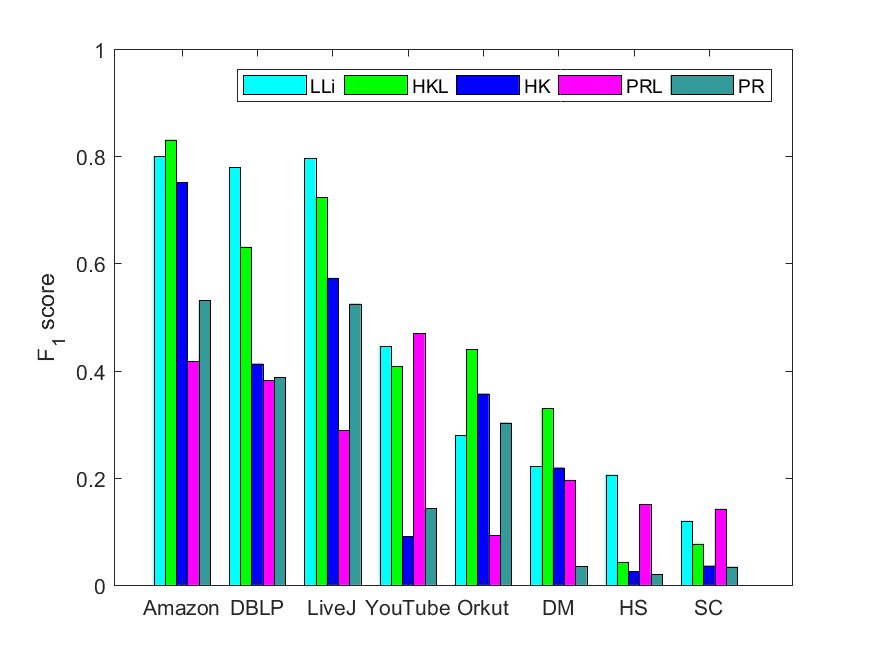}
	\vspace{-1em}
	\caption{Accuracy comparison with HKL and PRL on real-world networks (Community size truncated by local minimal conductance).}
	\label{Fig:realLocalCompare}
\end{figure}

\begin{figure}[!t]
	\vspace{-1em}
	\centering
	\subfigure[\bf{LFR\_s\_0.1}]{
		\includegraphics[width=2.65in]{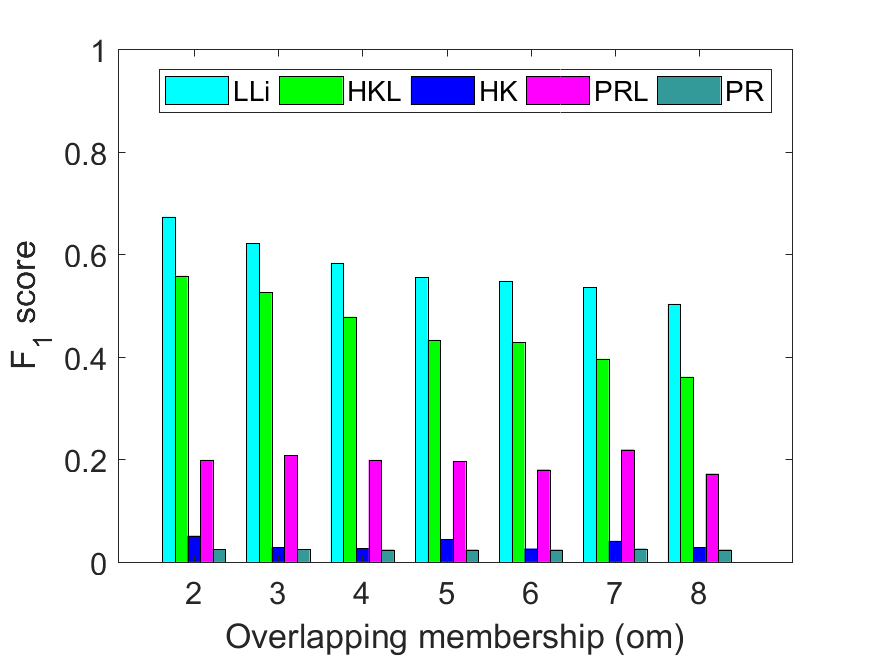}}
		\vspace{-1em}
	\subfigure[\bf{LFR\_s\_0.5}]{
		\includegraphics[width=2.65in]{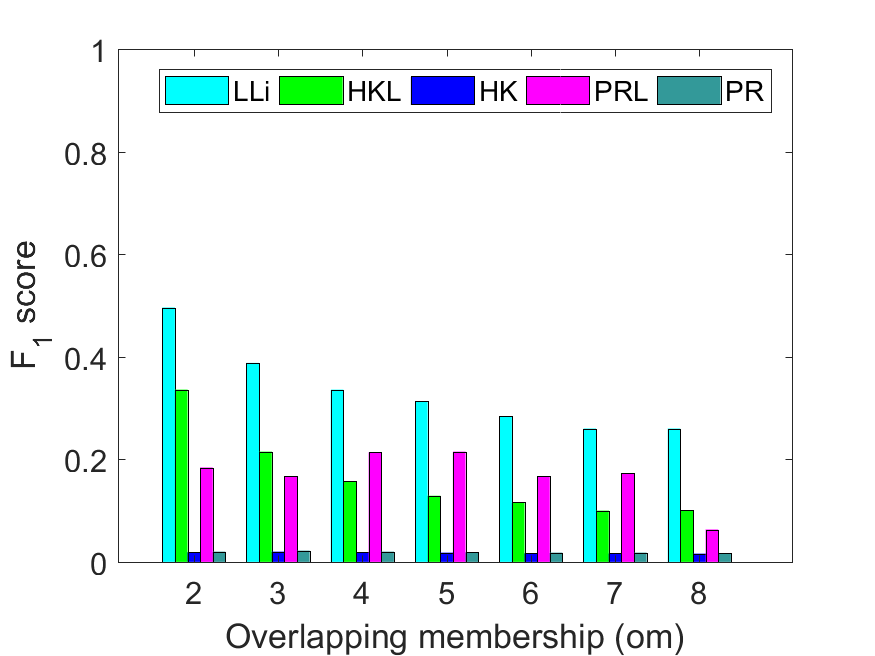}}
	\vspace{-1em}
	\subfigure[\bf{LFR\_b\_0.1}]{
		\includegraphics[width=2.65in]{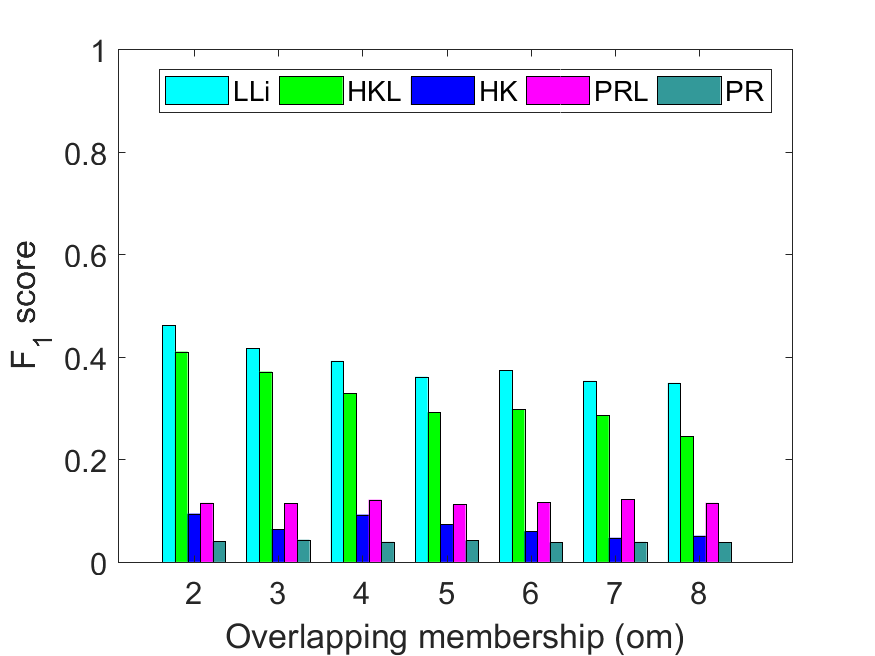}}
		\vspace{-1em}
	\subfigure[\bf{LFR\_b\_0.5}]{
		\includegraphics[width=2.65in]{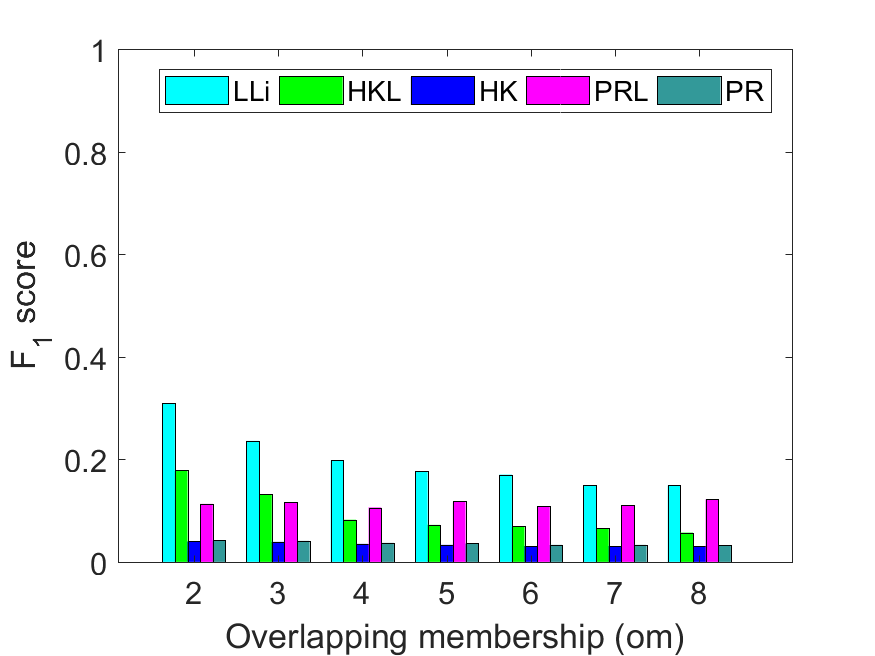}}
\vspace{1em}
	\caption{Accuracy comparison with HKL and PRL on LFR datasets (Community size truncated by local minimal conductance).}
	\label{Fig:LFRLocalCompare}
\end{figure}

\subsubsection{Comparison on LFR Datasets}
For each of the synthetic LFR datasets (Properties described in subsection \ref{subsec:LFR}. Recall that the ground truth communities are in size [10,50] for \textbf{s} and [20,100] for \textbf{b}), Table \ref{table: LFRcomparison} shows the average detection accuracy (The best value appears in bold in each row) and the average running time, Table \ref{table: LFRcomparison2} further shows the average community size and the average conductance of the detected communities.

For the four groups with a total of 28 LFR datasets, LLi clearly outperforms the state-of-the-art baselines, LEMON, PGDc-d, HK and PR, as evaluated by $F_1$ score in Table \ref{table: LFRcomparison}. On further analysis of each group, shown in Table \ref{table: LFRcomparison2}, we see that:

\begin{itemize}
	\item \textbf{LFR\_s\_0.1}. The size of the ground truth is in [10, 50], and the average conductance of the ground truth is 0.522. The average community size of LLi over the seven LFR\_s\_0.1 datasets is 20, while the baselines detect either much smaller communities in size 3 (PGDc-d) or much larger communities in size 2300 (HK and PR). Due to the difference on the detected community size, LLi has a stable conductance of around 0.62 while the baselines have a high conductance of around 0.95 (PGDc-d) or a low conductance of around 0.29 (HK and PR).
	\item \textbf{LFR\_s\_0.5}. The size of the ground truth is in [10, 50], and the average conductance of the ground truth is 0.746. The average community size of LLi is considerably larger in 145, and the average conductance of 0.798 is close to that of the ground truth. LEMON and PGDc-d achieve smaller communities with higher conductance. HK and PR find much larger communities, leading to apparently smaller conductance.
	\item \textbf{LFR\_b\_0.1}. The size of the ground truth is in [20, 100], and the average conductance of the ground truth is 0.497. The average community size of LLi is suitable in 28. The baselines detect much smaller communities of size around 3 (PGDc-d) or much larger communities of size around 2300 (HK and PR). Due to the difference on the detected community size, LLi has a stable conductance of around 0.70, while the baselines have either much higher conductance of around 0.98 (PGDc-d) or much lower conductance of around 0.30 (HK and PR).
	\item \textbf{LFR\_b\_0.5}. The size of the ground truth is in [20, 100], and the average conductance of the ground truth is 0.733.  The average community size of LLi is considerably larger in 182, and the average conductance of 0.808 is close to that of the ground truth. LEMON and PGDc-d achieve smaller communities with higher conductance. HK and PR find much larger communities, leading to apparently smaller conductance.
\end{itemize}

It is reasonable that the detection accuracy decays on graphs where there exist more overlapings indicated by higher $om$ and $on$. LLi is adaptive to find suitable size of communities for different configurations, and  substantially outperforms the baselines. LEMON and PGDc-d tend to find much-smaller-size communities with higher-conductance, while HK and PR tend to find much-larger-size communities with lower-conductance and their detection accuracy is very low in less than 0.1 on average on each of the four groups of datasets.

For the running time, as shown in Table \ref{table: LFRcomparison}, LLi is fast in 1 to 2 seconds while the four baselines, LEMON, PGDc-d, HK and PR, are slightly faster in less than 1.5 seconds. This may due to the implementation difference on programming language.

\subsubsection{More Comparison on the variants of HK and PR}
In order to evaluate the effectiveness of using the local minimal conductance for determining the community boundary, we compare LLi with the modification versions of HK and PR, based on local minimal conductance rather than global minimal conductance, denoted by HKL and PRL.

Fig. \ref{Fig:realLocalCompare} and Fig. \ref{Fig:LFRLocalCompare} illustrate the average detection accuracy on real-world networks and LFR datasets, respectively. For real-world networks, Fig. \ref{Fig:realLocalCompare} shows that HKL yields higher accuracy on each dataset compared with HK, while PRL achieves lower accuracy compared with PR on Amazon, DBLP, LiveJ and Orkut. On average, LLi still outperforms HKL and PRL. For LFR datasets, Fig. \ref{Fig:LFRLocalCompare} shows that LLi outperforms HKL and PRL on each LFR network. HKL and PRL yield considerably higher accuracy on each LFR dataset compared with HK and PR, respectively. These experiments show that it is useful for improving performance based on local minimal conductance truncation.

\subsubsection{Summary on the Proposed Algorithm}
For all comparisons of LLi, LOSP with light lazy, with four baselines, we see that LLi is adaptive for different configurations of the synthetic LFR datasets. LLi could find small communities with dozens of members for \textbf{s}: [10,50] and medium communities with hundreds of members for \textbf{b}: [20,100].

For well-defined community structure (low conductance, less overlapping, reasonably small), LLi has a high detection accuracy (\textbf{LFR\_s\_0.1}, \textbf{LFR\_b\_0.1}). When the overlapping membership ($om$) or the number of overlapping nodes ($on$) increases, the community structure becomes more mixed, and the detection accuracy decays. Nevertheless, LLi always yields the best accuracy for different parameter settings. By comparison, LEMON always finds communities of size around 20, PGDc-d always finds communities of size less than 10, and HK as well as PR always find large communities of size around 2300. Thus, they are not very scalable to networks with diverse community structure.

Comparisons on real-world datasets also show that LLi is very good for finding small communities with reasonably low conductance (Amazon, DBLP and LiveJ). If the community structure is not very clear, for example, with high conductance of around 0.90 on the Biology datasets, LLi tends to find larger communities using longer time due to the sweep search, and it does not yield the best accuracy.

Therefore, our method can be used for large-scale real-world complex networks, especially when the communities are in reasonable size of no greater than 500, and the network has a reasonable clear community structure.

\section{Conclusion}
This paper systematically explores a family of local spectral methods (LOSP) for finding members of a local community from a few randomly selected seed members. Based on a Krylov subspace approximation, we define ``approximate eigenvectors'' for a subgraph including a neighborhood around the seeds, and describe how to extract a community from these approximate eigenvectors.  By using different seed sets that generate different
subspaces, our method is capable of finding overlapping communities. Variants of LOSP are introduced and evaluated. Four types of random walks with different diffusion speeds are studied, regular random walk and inverse random walk are compared, and analysis on the link between Krylov subspace and eigenspace is provided.  For this semi-supervised learning task, LOSP outperforms prior state-of-the-art local community detection methods in social and biological networks as well as synthetic LFR datasets.

\appendix
\section{The proof of theorem 4.1}
\label{sec:Bounding}
Before given the proof of Theorem 4.1, we first give the following theorem.

\begin{theorem}(Courant-Fischer Formula)\label{CFF}
Let $\mathbf{H}$ be an $n\times n$ symmetric matrix with eigenvalues $\lambda_1^{(\mathbf{H})} \leq \lambda_2^{(\mathbf{H})} \leq \ldots \leq \lambda_n^{(\mathbf{H})}$  and corresponding eigenvectors $\mathbf{v}_1, \mathbf{v}_2, \ldots, \mathbf{v}_n$. Then

\begin{equation}
\begin{aligned}
\lambda_1^{(\mathbf{H})} & = \min_{\Vert\bf{x}\Vert_2 = 1} \bf{x^T}\bf{H}\bf{x} = \min_{\bf{x}\neq \bf{0}} \frac{\bf{x^T}\bf{H}\bf{x}}{\bf{x^T}\bf{x}}, \notag
\\
\lambda_2^{(\mathbf{H})} & = \min_{\Vert\bf{x}\Vert_2 = 1 \atop \bf{x}\perp \mathbf{v}_1} \bf{x^T}\bf{H}\bf{x} = \min_{\bf{x}\neq \bf{0}\atop \bf{x}\perp \mathbf{v}_1} \frac{\bf{x^T}\bf{H}\bf{x}}{\bf{x^T}\bf{x}},  \notag
\\
\lambda_n^{(\mathbf{H})} & = \max_{\Vert\bf{x}\Vert_2 = 1} \bf{x^T}\bf{H}\bf{x} = \max_{\bf{x}\neq \bf{0}} \frac{\bf{x^T}\bf{H}\bf{x}}{\bf{x^T}\bf{x}}. \notag
\end{aligned}
\end{equation}
\end{theorem}

We will not include the proof of the Courant-Fischer Formula here. The interested reader is referred to \cite{golub2012matrix}.

Let $\mathbf{L_{sym}}$ be the normalized graph Laplacian matrix of $G_s$ with eigenvalues $\lambda_1 \leq \lambda_2 \leq ... \leq \lambda_{n_s}$ and corresponding eigenvectors $\mathbf{q}_1, \mathbf{q}_2, \ldots, \mathbf{q}_{n_s}$. According to Proposition 3 of \cite{Spectral2007}, $\mathbf{L_{sym}}$ has $n_s$ non-negative eigenvalues.
As $\mathbf{L_{sym}(D_s^{\frac{1}{2}}e)} = \bf{0}$ where $\mathbf{e}$ is the vector of all ones, we have $\lambda_1=0$, $\mathbf{q}_1=\mathbf{\frac{D_s^{\frac{1}{2}}e}{\Vert D_s^{\frac{1}{2}}e \Vert_2}}$.

By Theorem \ref{CFF}, we have
\begin{equation}
\label{Eq:lamda}
\lambda_2  = \min_{\bf{x} \neq \bf{0} \atop \bf{x}\perp \mathbf{q}_1} \frac{\bf{x^T}\bf{L_{sym}}\bf{x}}{\bf{x^T}\bf{x}}
 = \min_{\bf{z} \neq \bf{0} \atop \bf{z}\perp \bf{D_se}} \frac{\bf{z^T}\bf{L}\bf{z}}{\bf{z^T}D_s\bf{z}}
 = \min_{\bf{z} \neq \bf{0} \atop \bf{z}\perp \bf{D_se}} \frac{\sum\limits_{i\thicksim j}(z_i-z_j)^2}{\sum_{i}d_i{z_i}^2},
\end{equation}
where $\mathbf{z} = \mathbf{D_s^{-\frac{1}{2}}x}$, $d_i$ is the degree of the $i$th node and $\sum\limits_{i\thicksim j}$ denotes the sum over all unordered
pairs \{$i,j$\} for which $i$ and $j$ are adjacent.
\begin{proof}[\bf{Proof of Theorem 4.1.}]
Let $\bf{z} = \bf{y} - \sigma \bf{e},$
where $\mathbf{y} \in \{0,1\}^{n_s\times 1}$ is a binary indicator vector representing community $C$ in graph $G_s$, $\mathbf{e}$ the vector of all ones, and
$\sigma = \frac{\text{vol}(C)}{\text{vol}(V_s)}$.

We can check that $\bf{z}\perp \bf{D_se}$:
\begin{equation}
\mathbf{z^TD_se}  = \mathbf{y^TD_se-\sigma e^TD_se}
= \text{vol}(C) - \frac{\text{vol}(C)}{\text{vol}(V_s)}\text{vol}(V_s) = 0. \notag
\end{equation}

We also know
\[
\bf{z^TLz} = \bf{(\bf{y} - \sigma \bf{e})^TL(\bf{y} - \sigma \bf{e})} = \bf{y^TLy}.
\]

It remains to compute
\begin{equation}
\begin{aligned}
\bf{z^TD_sz} & = \bf{(\bf{y} - \sigma \bf{e})^TD_s(\bf{y} - \sigma \bf{e})}  \\
           & = \mathbf{y^TD_sy} -2\sigma \mathbf{y^TD_se} + \sigma^2\mathbf{e^TD_se}  \\
           & = \text{vol}(C)-2\sigma\text{vol}(C)+\sigma^2\text{vol}(V_s) \\
           & = \frac{\text{vol}(C)\text{vol}(V_s-C)}{\text{vol}(V_s)}. \notag
\end{aligned}
\end{equation}
By Eq. (\ref{Eq:lamda}), we have
\begin{equation}
\label{Eq:lamda2}
\lambda_2 \leq \frac{\bf{z^TLz}}{\bf{z^TD_sz}} = \frac{\mathbf{y^TLy} \cdot \text{vol}(V_s)}{\text{vol}(C)\text{vol}(V_s-C)}.
\end{equation}
As the larger value of $\text{vol}(C)$ and $\text{vol}(V_s-C)$ is at least half of $\text{vol}(V_s)$,
\begin{equation}
\lambda_2 \leq 2\frac{\bf{y^TLy}}{\text{min}(\text{vol}(C),\text{vol}(V_s-C))} = 2\frac{\bf{y^TLy}}{\text{vol}(C)} = 2\frac{\bf{y^TLy}}{\bf{y^TD_sy}}  = 2\Phi(C). \notag
\end{equation}
Therefore,
\[
\frac{\lambda_2}{2} \leq \Phi(C).
\]
And it is obvious that
\begin{equation} \label{Eq:Lamda22}
\Phi(C) = \frac{\bf{y^TLy}}{\bf{y^TD_sy}} \leq 1.
\end{equation}
\end{proof}

%
%
%
%


\begin{acks}
This work was supported by National Natural Science Foundation of China (61772219, 61702473, 61572221).
\end{acks}

\bibliographystyle{ACM-Reference-Format}
\bibliography{LOSP}

\end{document}